\RequirePackage[2020-10-01]{latexrelease}
\documentclass[MPS,Times2COL]{WileyNJDv5}

\usepackage{algorithm,algpseudocode}
\usepackage{bm}
\usepackage{framed}
\usepackage{comment}
\usepackage{subcaption}

\articletype{ORIGINAL RESEARCH}%

\received{Date Month Year}
\revised{Date Month Year}
\accepted{Date Month Year}
\journal{Journal}
\volume{00}
\copyyear{2026}
\startpage{1}

\begin{document}

\title{Fast Convergence and Robustness for Two-Layered Forgetting Recursive Least Square under Finite Excitation}

\author[1]{Satoshi Tsuruhara}
\author[2]{Kazuhisa Ito}

\authormark{S. Tsuruahra \textsc{et al.}}
\titlemark{Fast Convergence and Robustness for Two-Layered Forgetting Recursive Least Square under Finite Excitation}

\address[1]{\orgdiv{Postdoctoral Fellow, College of Systems Engineering and Science}, \orgname{Shibaura Institute of Technology}, \orgaddress{\state{307 Fukasaku, Minuma, Saitama 3378570}, \country{Japan}}}

\address[2]{\orgdiv{Mechatronics Program, College of Systems Engineering and Science}, \orgname{Shibaura Institute of Technology}, \orgaddress{\state{307 Fukasaku, Minuma, Saitama 3378570}, \country{Japan}}}

\corres{Corresponding author: Satoshi Tsuruhara. \email{nb23110@shibaura-it.ac.jp}}

\presentaddress{307 Fukasaku, Minuma, Saitama 3378570, Japan}


\abstract[Abstract]{Under nonpersistent excitation (non-PE) conditions, conventional methods such as exponential forgetting (EF) or directional forgetting (DF) recursive least squares (RLS) that rely on direct regressor vectors exhibit inherent limitations in terms of stability guarantees for parameter errors, robustness to system changes, and convergence rates. To address these limitations, this study introduces a novel two-layer forgetting RLS (TLF-RLS) identification method based on an augmented regressor matrix constructed using DF, which ensures global exponential stability and enhances robustness under non-PE condition. However, the convergence rate of the parameter is strongly dependent on the forgetting factor because of the introduction of EF in the outer layer, which causes an estimation windup under non-PE condition. To address this issue, a novel reconfiguration-based EF (ReEF) algorithm is proposed, which is achieved through variable- and matrix-based forgetting related to the magnitude of the eigenvalues of the current covariance matrix. Theoretical analysis indicates that TLF-RLS with ReEF algorithm guarantees uniform ultimate boundedness of the condition number under mild assumptions. Consequently, the proposed method resolves the trade-off between fast parameter convergence and robustness in both transient and steady-state responses under changes in system characteristics. Numerical simulations of three aforementioned cases demonstrate the effectiveness of the proposed method.}


\jnlcitation{\cname{%
\author{Satoshi T.},
\author{Kazuhisa I.}}.
\ctitle{Fast Convergence and Robustness for Two-Layered Forgetting Recursive Leas Squares under Finite Excitation}
\cvol{2021;00(00):1--18}.}

\maketitle

\renewcommand\thefootnote{}

\renewcommand\thefootnote{\fnsymbol{footnote}}
\setcounter{footnote}{1}

\section{Introduction}\label{sec1}
Adaptive identification has long been considered a key problem, wherein the set of parameter of a given mathematical model structure is estimated online \cite{Ljung1999}.
In particular, ensuring the convergence of true parameter values is an essential objective. For this purpose, a persistent excitation (PE) is necessary and sufficient condition for guaranteeing convergence to true parameter values in normalized gradient (NG) and recursive least squares (RLS) identification methods with exponential forgetting (EF) \cite{tao,landau2011adaptive}. 
However, in practice, it is difficult to satisfy the PE condition. For example, excessive vibration may occur when quasi-white noise such as pseudo-random binary sequences is applied to mechanical systems, leading to system failure. Furthermore, in adaptive control systems, the reference trajectory is typically specified by the designer, and the control input is generated accordingly, making it difficult to ensure the PE condition.
Recently, several adaptive identification and control methods that relax the PE condition while still guaranteeing stability and true value convergence have been proposed \cite{survey_PE1, survey_PE2}. The potential benefits of adaptive control, as exemplified by approaches such as concurrent learning (CL) \cite{CL1,CL2,DCL,DCL2}, composite model reference adaptive control \cite{CL-composite}, and dynamic regressor extension and mixing \cite{DREM}, have been reported. 
In particular, CL improves on the conventional NG algorithm by adding a recorded data term that accumulates important past input/output (I/O) data, enabling the easy convergence of parameters to true values under a finite excitation (FE) condition.
The FE condition is weaker than the PE condition because the former requires excitation only over a certain interval.
However, the CL method has two practical issues, namely, it cannot handle changes in characteristics including parameter jumps and the estimated parameter has a very slow convergence rate.

In this study, we propose an identification method based on effective forgetting factors to address these issues.
Forgetting factor algorithms have been studied to improve the parameter convergence rate in RLS and enhance robustness against changes in system characteristics \cite{FF_survey}. However, it is known that the widely used EF algorithm causes estimation windup under non-PE condition because it cannot guarantee the existence of an upper bound on the covariance matrix. Consequently, EF-based RLS (EF-RLS) \cite{EF-RLS} has easily become unstable \cite{DF}. In contrast, several methods have been proposed to ensure the boundedness of the covariance matrix under non-PE condition. Among these, we focus on the directional forgetting (DF) algorithm proposed by Cao et al. \cite{DF}. Inspired by reference \cite{DF-CL}, we propose a CL method based on DF (DF-CL) for discrete-time systems.
Although the DF-CL method is considered effective in addressing these issues under the FE condition, and its convergence rate is relatively slow because it is based on the NG method. In previous studies, extensions based on finite- and fixed-time frameworks, which guarantee the finite-time convergence of the estimated parameters, have been proposed \cite{FT_CL1,FT_CL2,FT_CL3}. These approaches require only minor modifications to the parameter update law and improve convergence to true parameter values. However, the improvement remains limited because the essential characteristics of NG algorithm is preserved.

To address these problems, this study proposes a novel two-layer forgetting RLS (TLF-RLS) identification method that retains the advantages of the DF-CL and EF-RLS \cite{TLF-RLS}. The proposed method has a two-layer structure, wherein the inner layer generates an augmented regressor matrix based on DF and the outer layer estimates the parameter based on EF. Consequently, convergence to the true parameter values under FE conditions with faster parameter convergence than those of CL-based methods and high robustness against parameter jumps are achieved.
In related work, several RLS-based identification methods have been proposed to guarantee convergence to the true parameter values without requiring the PE condition \cite{RLS_FE1,RLS_FE2,RLS_FE3}. In \cite{RLS_FE1}, the finite-time convergence of the parameter error is achieved by appropriately introducing a regularization term into information matrix. Reference \cite{RLS_FE2} guaranteed the exponential convergence of the parameter error by modifying the regressor vector and parameter update law to ensure the positive definiteness of the information matrix. Reference \cite{RLS_FE3} guaranteed the asymptotic convergence of the parameter error by improving the conventional CL method.
Although these methods are effective, \cite{RLS_FE1} and \cite{RLS_FE3} guarantee only asymptotic convergence and do not include forgetting factors, limiting their applicability to time-invariant systems. In \cite{RLS_FE2}, continuous-time systems are considered and the EF algorithm is applied to regressor updates. Therefore, it is difficult to achieve satisfactory performance unless the covariance matrix is appropriately updated and its thresholds appropriately selected.
The proposed method has a two-layer structure that effectively combines the two forgetting factors, significantly improving the convergence rate of the parameter while maintaining high robustness against changes in system characteristics.
Further, this paper discusses the estimation windup problem \cite{DF} in the TLF-RLS method. Conventionally, the DF algorithm guarantees the boundedness of the covariance or information matrix. Estimation windup can be suppressed because the lower bound of the information matrix can be ensured by both the value of the forgetting factor and the level of signal excitation.
Although the TLF-RLS method utilizes the DF algorithm to satisfy the PE condition and employs the EF algorithm to ensure a fast parameter convergence rate, it cannot necessarily ensure that this lower bound of the information matrix remains sufficiently large. Consequently, the selection of the forgetting factor in TLF-RLS is critical, which makes its optimal design difficult.
Consequently, although TLF-RLS maintains robustness against changes in system characteristics in the steady-state response and achieves an improved convergence rate, it may not sufficiently suppress the estimation windup for such parameter jumps. To address this issue, this paper proposes reconfiguration exponential forgetting (ReEF), inspired by the concept of variable directional forgetting (VDF) \cite{VDF1,VDF2}. ReEF is a modification of the EF algorithm and is used in the outer layer forgetting factor.
The proposed outer layer forgetting factor guarantees the uniform ultimate boundedness of the condition number under mild assumptions.
Therefore, the condition number of the information or covariance matrix is improved, and we can expect to suppress the estimation windup while maintaining the convergence rate of the parameters.
In summary, under non-PE condition, conventional single-layer parameter adaptive identification methods such as RLS based on EF, DF, and VDF, as well as the DF-CL method cannot simultaneously achieve convergence to the true parameter values, fast convergence rates, robustness against changes in system characteristics, and suppression of estimation windup. In contrast, the proposed method can achieve these properties through its two-layer structure using several forgetting factors under FE conditions.

The main contributions of this study are as follows:
1) We propose a novel TLF-RLS identification method that combines DF and EF, and we demonstrate its global exponential stability with respect to parameter estimation errors under FE conditions.
2) We demonstrate the effectiveness of the proposed method through numerical examples, including comparisons with various existing methods.
3) We propose a novel outer layer forgetting factor for TLF-RLS that suppresses the estimation windup by guaranteeing the uniform ultimate boundedness of the condition number of the covariance matrix.
4) We demonstrate through numerical examples with parameter jumps that robustness against changes in system characteristics is improved.
The remainder of this paper is organized as follows: Section 2 presents the preliminaries and the problem formulation. Section 3 describes the three proposed methods and highlights their relationship with conventional methods. Section 4 briefly describes the stability analysis of the proposed method. Section 5 validates the effectiveness of the proposed method through several numerical simulations.
 Finally, Section 6 presents the conclusions and future work.

\textit{Notation}$\ $
Throughout this paper, for an invertible matrix $A$, $A^{-T}\triangleq (A^{-1})^T=(A^T)^{-1}$. $\|\cdot\|$ represents the Euclidean norm for vectors. $\mathbb{R}_{>0}$ represents the positive real space. For a symmetric matrix $A\in\mathbb{R}^{n\times n}$,
$\sigma_{\max}(A),\ \sigma_{\min}(A)$ represent the maximum and minimum eigenvalues of $A$, respectively. Both $\lambda$ and $\mu$ represent the forgetting factors. In addition, the condition number of $A$ is defined as
$\kappa(A)\triangleq\frac{\sigma_{\max}(A)}{\sigma_{\min}(A)}$.

\section{Problem Formulation and Preliminaries}
This section presents the formulation of the adaptive parameter identification with objectives, and the excitation conditions.
\subsection{Problem Formulation}
Consider a linear parametric model given as
\begin{align}\label{ma:LPM}
y(k+1)=\phi^T(k)\theta,
\end{align}
where $y(k)\in\mathbb{R}$, $\phi(k)\in\mathbb{R}^n$, and $\theta\in\mathbb{R}^n$ represent the output, regressor vector, and unknown parameter vector, respectively.
The prediction model can be expressed using the predicted output $\hat{y}(k)\in\mathbb{R}$ as
\begin{align}\label{ma:ELPM}
\hat{y}(k+1)=\phi^T(k)\hat{\theta}(k),
\end{align}
where $\hat{\theta}(k)\in\mathbb{R}^n$ represents the estimate of $\theta$.
Subsequently, the identification error $q(k+1)\in\mathbb{R}$ is defined as 
\begin{align}\label{ma:id_error}
q(k+1)&\triangleq\hat{y}(k+1)-y(k+1)\\
&=\phi^T(k)\hat{\theta}(k)-\phi^T(k)\theta = \phi^T(k)\tilde{\theta}(k),
\end{align}
where $\tilde{\theta}(k)\triangleq \hat{\theta}(k)-\theta\in\mathbb{R}^n$ represents the parameter error.
The objective of adaptive parameter identification proposed in this study is to satisfy the condition
\begin{align}
\lim_{k\to\infty}\|\hat{\theta}(k)-\theta\|=0\quad (\mathrm{exponentially}).
\end{align}
This condition indicates that the parameters converge exponentially to their true values. The second objective is to ensure that the parameter error remains sufficiently small over the entire interval when the true parameter values change, such as in the case of parameter jumps.

\subsection{Preliminaries}
The excitation conditions considered in this study are summarized below. In both the NG and EF-RLS methods, parameter convergence to the true values is guaranteed if and only if the PE condition is satisfied \cite{tao}.
\begin{definition}\label{def:PE}
The bounded matrix signal $\Phi(k)\in\mathbb{R}^{n\times n}$, $n\geq 1$, is PE if there exist $\delta>0$ and $\alpha_{\phi}>0$ such that
\begin{align}
\sum_{k=\sigma}^{\sigma+\delta}\Phi(k)\Phi^T(k)\geq\alpha_{\Phi}I,\ \forall \sigma\geq k_0.
\end{align}
\end{definition}
\begin{definition}\label{def:FE}
The bounded matrix signal $\Phi(k)\in\mathbb{R}^{n\times n}$, $n\geq 1$ is exciting over time sequence set $\{\sigma_0,\sigma_0+1,\ldots,\sigma_0+\delta_\Phi\},\ \delta_\Phi>0,\sigma_0\geq k_0$, if for some $\alpha_\Phi>0$, it holds that
\begin{align}
\sum_{k=\sigma_0}^{\sigma_0+\delta_\Phi}\Phi(k)\Phi^T(k)\geq \alpha_{\Phi}I
\end{align}
This condition is known as FE.
\end{definition}
As mathematical preliminaries, these two definitions of the excitation conditions above for the regressor are introduced based on \cite{tao,Slotine1991-ANC,TLF-RLS}. These excitation conditions are discussed in a matrix form, particularly in the case of TLF-RLS.
For the FE condition to hold, condition \ref{cond:rank} must be satisfied. This condition is formulated in the context of the regressor vector employed in the DF-CL method.
\begin{condition}[rank condition]\label{cond:rank}
The signal matrix comprising the summations of the regressor vector $\phi(k)$ (e.g., augmented regressor matrix), is a full column rank matrix.
\end{condition}

\section{Proposed methods}
We describe the generation of the augmented regressor matrix that accumulates input/output (I/O) data using the DF algorithm as the foundation for all proposed methods. Next, we present the three parameter update laws under the FE condition based on this algorithm and clarify their relationships with the conventional methods.

\subsection{Generation of the augmented regressor matrix}
We generate an augmented regressor matrix using the DF algorithm to achieve the continued satisfaction of the rank condition (See Condition \ref{cond:rank}) and boundedness of the information matrix, and for improved robustness against changes in the system characteristics.
Let the accumulated data of the regressor $\phi(k)$ be denoted by the augmented regressor matrix $\Phi(k)\in\mathbb{R}^{n\times n}$, and let the accumulated output data be denoted by the auxiliary vector $X(k)\in\mathbb{R}^{n}$. In the absence of forgetting, these signals are updated as
\begin{align}\label{ma:acc}
\begin{cases}
\Phi(k+1)=\Phi(k)+\dfrac{\phi(k)\phi^T(k)}{m^2(k)}\\[10pt]
X(k+1)=X(k)+\dfrac{\phi(k)y(k+1)}{m^2(k)}
\end{cases},
\end{align}
where $m(k)\triangleq \sqrt{1+\phi^T(k)\phi(k)}$.
Based on \eqref{ma:acc}, we define the extended identification error as
\begin{align}
\Phi(k)\hat{\theta}(k) - X(k)
= \Phi(k)\hat{\theta}(k) - \Phi(k)\theta
= \Phi(k)\tilde{\theta}(k).
\end{align}
As mentioned earlier, the update of accumulated data without the forgetting factor cannot consider changes in the system characteristics. Therefore, based on the DF algorithm, the update law is modified as 
\begin{align}\label{ma:acc_DF}
\begin{cases}
\displaystyle
\Phi(k+1)=\Phi(k)-\mu\frac{\Phi(k)\phi(k)\phi^T(k)}{\phi^T(k)\Phi(k)\phi(k)}\Phi(k)+\frac{\phi(k)\phi^T(k)}{m^2(k)}\\[10pt]
\displaystyle
X(k+1)=X(k)-\mu\frac{\Phi(k)\phi(k)\phi^T(k)}{\phi^T(k)\Phi(k)\phi(k)}X(k)+\frac{\phi(k)y(k+1)}{m^2(k)}
\end{cases},
\end{align}
where $\mu\in(0,1)$ denotes the forgetting factor of DF.
A value closer to $1$ results in a stronger forgetting, which is the opposite of the conventional forgetting factor. The update law accumulates data while prioritizing current data and retain less weighting on the past data by adding the second term based on DF to \eqref{ma:acc}.
However, in a general DF-RLS, initial matrices are set to be positive definite. In these proposed methods, the initial values of the augmented regressor matrix and auxiliary vector need to be set to zero to avoid bias in the accumulated data. Therefore, only the accumulation step, i.e., \eqref{ma:acc}, is performed until $\Phi(k)$ satisfies Condition \ref{cond:rank}. Once Condition \ref{cond:rank} is satisfied, positive definiteness is preserved, as indicated in Proposition \ref{prop_DF}. This is because the DF algorithm performs forgetting only in the direction of the newly added regressor data.
These algorithms are summarized in Algorithm \ref{alg:gene_DF}.
\begin{algorithm}
\caption{Signal generation by accumulation based on DF}\label{alg:gene_DF}
\begin{algorithmic}
\State \textbf{Step 1}: Set $\Phi(0)=0_{n\times n}$, $X(0)=\bm{0}_{n}$, $\mu\in(0,1)$
\State \textbf{Step 2}:
\If{$\displaystyle \mathrm{rank}(\Phi(k))<\mathrm{rank}\left(\Phi(k)+\frac{\phi(k)\phi^T(k)}{m^2(k)}\right)$}
   \State $\displaystyle \Phi(k+1)=\Phi(k)+\frac{\phi(k)\phi^T(k)}{m^2(k)}$
    \State $\displaystyle X(k+1)=X(k)+\frac{\phi(k)y(k+1)}{m^2(k)}$
\Else
    \State $\displaystyle \Phi(k+1)=\Phi(k)-\mu\frac{\Phi(k)\phi(k)\phi^T(k)}{\phi^T(k)\Phi(k)\phi(k)}\Phi(k)+\frac{\phi(k)\phi^T(k)}{m^2(k)}$
    \State $\displaystyle X(k+1)=X(k)-\mu\frac{\Phi(k)\phi(k)\phi^T(k)}{\phi^T(k)\Phi(k)\phi(k)}X(k)+\frac{\phi(k)y(k+1)}{m^2(k)}$
\EndIf
\State \textbf{Step 3}: Set $k \leftarrow k+1$ and go to \textbf{Step 2}
\end{algorithmic}
\end{algorithm}

\begin{proposition}\label{prop_DF}
Assume that the regressor vector is bounded and satisfies the FE condition. Then, for the augmented regressor matrix $\Phi(k)$ generated by Algorithm \ref{alg:gene_DF}, there exist $\alpha_{\phi}, \beta_{\phi}\in\mathbb{R}_{>0}$ such that for all time steps $k \ge k_e$,
\begin{align}
\alpha_{\phi}I < \Phi(k) < \beta_{\phi}I,
\end{align}
where $k_e$ represents the time step at which $\Phi(k)$ becomes positive definite for the first time.
\end{proposition}
\begin{proof}
The regressor vector is bounded and satisfies the FE condition, and therefore, the normalized regressor vector $\frac{\phi(k)}{m(k)}$ also satisfies the FE condition.
Then, we have
\begin{align}
\exists k_e\in\mathbb{R}_{>0}\ \mathrm{s.t.}\ \sum_{i=0}^{k_e}\frac{\phi(i)\phi^T(i)}{m(i)}>0
\end{align}
Hence, under FE conditions, time step $k_e$ exists.
Therefore, based on the stability analysis of conventional DF-RLS \cite{DF}, the boundedness of the augmented regressor matrix is guaranteed.
\end{proof}
\begin{remark}
In Algorithm \ref{alg:gene_DF}, the augmented regressor matrix $\Phi(k)$ can be generated using EF instead of DF. However, in this case, the boundedness of the augmented regressor matrix $\Phi(k)$ cannot be guaranteed because data are uniformly forgotten under non-PE condition, which leads to a monotonically non-increasing in the eigenvalues of the augmented regressor matrix in directions where no new data are added. Hence, the stability of the identification method for the augmented regressor matrix using EF cannot be ensured under the FE condition.
\end{remark}

\subsection{DF-CL}
Based on the generation of the augmented regressor matrix $\Phi(k)$ described in the previous section, the modified parameter update law for the discrete-time CL method proposed by \cite{DCL,DCL2} is given as 
\begin{align}\label{ma:DF-CL}
\hat{\theta}(k+1)
= \hat{\theta}(k)
- \eta(k)\frac{\phi(k)q(k+1)}{m^2(k)}
- \eta(k)\bigl(\Phi(k)\hat{\theta}(k) - X(k)\bigr),
\end{align}
where $\eta(k)$ represents the learning weight and is selected from the range
\begin{align}
0 < \eta(k) < \frac{2 m^2(k-1)}{2 \lambda_{\max}\!\left(\phi(k)\phi^T(k)\right) + \lambda_{\max}\!\left(\Phi(k)\right) m^2(k-1)}.
\end{align}
Here, the following selection is practically effective.
\begin{align}
\eta(k)=\frac{m^2(k-1)}{2 \lambda_{\max}\!\left(\phi(k)\phi^T(k)\right) + \lambda_{\max}\!\left(\Phi(k)\right) m^2(k-1)}.
\end{align}
The parameter update law \eqref{ma:DF-CL} modifies the conventional NG method by adding a third term based on accumulated data. The existence of an augmented identification error that preserves positive definiteness ensures convergence to the true parameter values under FE conditions.
Furthermore, updating the augmented regressor matrix in Algorithm \ref{alg:gene_DF} not only preserves positive definiteness but also increases the inverse of its condition number. In discrete-time CL \cite{DCL}, maximizing the inverse of the condition number with respect to the accumulated data improves the convergence rate. Therefore, the proposed method is robust against changes in system characteristics and exhibits improved convergence performance.

\subsection{TLF-RLS}
Similar to the DF-CL algorithm, we propose a novel TLF-RLS algorithm, which is RLS-based identification approach that utilizes the augmented regressor matrix $\Phi(k)$.
Figure \ref{fig:concept_TLF} illustrates the concept of TLF-RLS.
In the inner layer, the augmented regressor matrix $\Phi(k)$ is generated and updated using the DF algorithm and regressor vector $\phi(k)$ in Algorithm \ref{alg:gene_DF}. In contrast, in the outer layer, an RLS method for a multiple-input multiple-output (MIMO) system is constructed using the EF algorithm for the augmented regressor matrix. The forgetting algorithm in the inner layer is referred to as the inner forgetting factor $\mu\in(0,1)$, whereas that in the outer layer is referred to as the outer forgetting factor $\lambda\in(0,1)$.

A key feature of this algorithm is formulated as shown below.
\begin{theorem}\label{thm:PE}\cite{TLF-RLS}
Assume that the regressor vector $\phi(k)$ satisfies the FE condition.
Then, for the augmented regressor matrix $\Phi(k)$, $\delta = k_e >0$, $\alpha,\ \beta\in\mathbb{R}_{>0}$ exist such that
\begin{align}\label{ma:PE_Omega}
\alpha_{\Phi} I\leq \sum_{k=\sigma}^{\sigma+\delta}\Phi^2(k)\leq \beta_{\Phi} I,\ \forall \sigma\geq 0.
\end{align}
Therefore, the augmented regressor matrix $\Phi(k)$ satisfies the PE condition.
\end{theorem}
\begin{proof}
From Proposition \ref{prop_DF}, after the time step $k_e$, the augmented regressor matrix $\Phi(k)$ is positive definite such that the longest interval for the existence of the lower bound $\alpha$ is positive for $\delta=k_e$.
The upper bound is obtained immediately from Proposition \ref{prop_DF} as well.
\end{proof}
\noindent This result indicates that, if the regressor vector $\phi(k)$ satisfies the FE condition, the augmented regressor matrix $\Phi(k)$ satisfies the PE condition. This is a significant result because it enables the design of RLS identification and control methods under PE condition. However, only the boundedness of the square of the augmented regressor matrix $\Phi^2(k)$ can be guaranteed, and its exact bound cannot be explicitly determined.
\begin{figure}[t]
\centering
\includegraphics[width=\columnwidth]{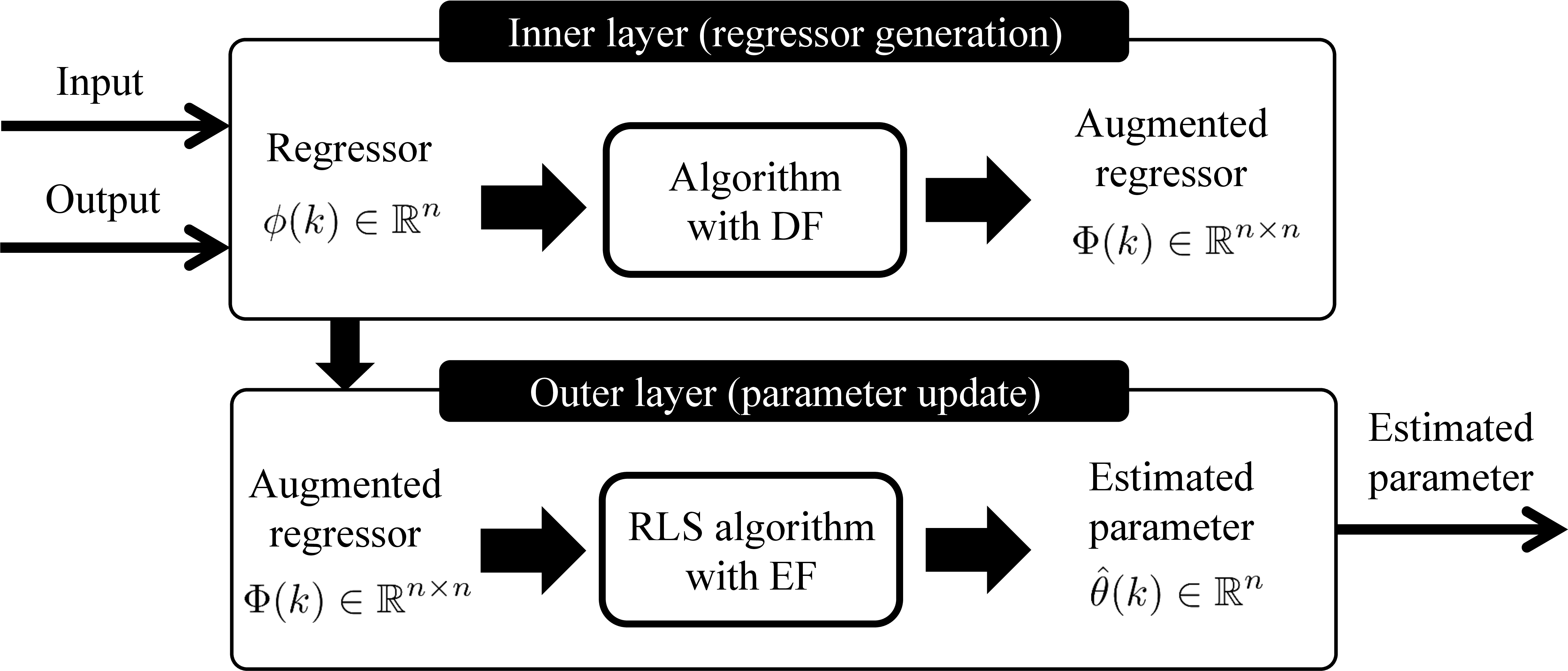}
\caption{Concept of TLF-RLS.}
\label{fig:concept_TLF}
\end{figure}
Subsequently, based on the augmented regressor matrix, we consider the cost function
\begin{align}\label{ma:cost_TLF}
J(k)
&=\sum_{i=1}^{k}\lambda^{k-i}
\bigl(\Phi(i)\hat{\theta}(k)-X(i)\bigr)^T
\bigl(\Phi(i)\hat{\theta}(k)-X(i)\bigr)\notag\\
&\quad+\lambda^{k}\bigl(\hat{\theta}(k)-\theta(0)\bigr)^T
R(0)\bigl(\hat{\theta}(k)-\theta(0)\bigr)
\end{align}
where $\lambda\in(0,1)$ represents the value of the forgetting factor referred to as the outer forgetting factor.
Then, for all $k\geq 0$, by minimizing the cost function \eqref{ma:cost_TLF} with respect to the estimated parameter $\hat{\theta}(k)$, we have
\begin{align}
\hat{\theta}(k+1)
&=\hat{\theta}(k)
- P(k)\Phi(k)N^{-1}(k)\bigl(\Phi(k)\hat{\theta}(k)-X(k)\bigr)
\label{ma:law_theta}\\
P(k+1)
&=\frac{1}{\lambda}\left(P(k)-P(k)\Phi(k)N^{-1}(k)\Phi(k)P(k)\right)
\label{ma:law_P}
\end{align}
where $P(k)\in\mathbb{R}^{n\times n}$ and $R(k)\in\mathbb{R}^{n\times n}$ represent the covariance and information matrices, respectively.
Given that $R(k)\triangleq P^{-1}(k)$, we set $P(0)=\gamma I,\ \gamma>0$.
Furthermore, $N(k)\triangleq \lambda I + \Phi(k)P(k)\Phi(k)$.

\subsection{TLF-RLS with ReEF}
Theorem \ref{thm:PE} ensures only the PE condition and does not improve the condition number of the covariance matrix. For robustness against changes in system characteristics, the minimum eigenvalue of the covariance matrix should be large, while the maximum eigenvalue should be small to avoid an estimation windup. Therefore, improving the condition number of the covariance matrix is essential for enhancing the parameter estimation performance.
The update of the information matrix for the TLF-RLS method can be expressed as 
\begin{align}\label{ma:R_temp}
R(k+1) = \lambda R(k) + \Phi^2(k)
\end{align}
Based on the previous discussion, the second term in \eqref{ma:R_temp} ensures positive definite. Thus, the positive definiteness of $R(k+1)$ is also preserved. The second term becomes a relatively small condition number because it is updated using the DF algorithm.
Hence, the primary factor causing a significant deterioration in the condition number of the information matrix is the first term based on the EF algorithm. A uniform forgetting by the outer forgetting factor accelerates parameter convergence; however, it can also rapidly degrade the condition number.
In contrast, the DF algorithm is not suitable as an outer forgetting factor because the DF-RLS method is only uniformly stable even under PE condition and cannot guarantee convergence to the true parameter values.
To address this issue, we propose a novel reconfiguration EF (ReEF) as the outer forgetting factor, which applies stronger forgetting in directions associated with larger data values and weaker forgetting in directions associated with smaller data values, thereby achieving a small condition number.
First, we must extend the forgetting factor to the matrix form.
Referring to \cite{FF_matrix} and noting that the regressor matrix is square, we propose the following RLS formulation with a matrix-forgetting factor via an augmented regressor matrix:,
\begin{align}
\hat{\theta}(k+1)
&= \hat{\theta}(k)- L(k)\Phi(k)\bar{N}^{-1}(k)
\bigl(\Phi(k)\hat{\theta}(k) - X(k)\bigr),\label{ma:ReEF_theta}\\
P(k+1)
&= L(k)
- L(k)\Phi(k)\bar{N}^{-1}(k)\Phi(k)L(k),\label{ma:ReEF_P}\\
L(k+1)
&= B(k+1)P(k+1)B^T(k+1),\label{ma:ReEF_L}
\end{align}
where $\bar{N}(k)\triangleq I + \Phi(k)L(k)\Phi(k)$, $B(k)\in\mathrm{R}^{n\times n}$ represents the forgetting matrix. Furthermore, $L(k)$ represents the modified covariance matrix, whose symmetry is preserved by multiplication with $B(k)$ on both sides.
The proposed method suppresses the estimation windup by appropriately designing $B(k)$. Considering the direction of data addition to improve the condition number, unlike DF or VDF, is not necessary because the outer layer updates the information matrix by adding a regressor matrix that always satisfies the PE condition.
Therefore, we focus on the covariance matrix prior to the update and evaluate the strength of available data using singular value decomposition, i.e., $P(k)=U(k)\Sigma_p(k)U^T(k)$, where $U(k)\in\mathbb{R}^{n\times n}$ represents the orthogonal matrix and $\Sigma_p(k)=\mathrm{diag}\{p_1,\ p_2,\ \ldots,\ p_n\}$, where $p_i>0\ (i=1,\ldots,n)$.
We design an eigenvalue matrix $\Lambda(k)$ and construct a forgetting matrix $B(k)=U(k)\Lambda^{-1}(k)U^T(k)$.
As discussed previously, the forgetting factor is tuned according to the strength of the covariance matrix data. In this context, we set the minimum forgetting factor $\lambda_{\min}(k)$ and introduced a design parameter $0 < \rho < 1$ related to the decrease in the condition number.
The selection of this parameter as a part of the stability analysis is discussed in the next section.
In contrast, $\lambda_{\max}(k)$ is updated as 
\begin{align}
\lambda_{\max}(k)= \min\{\underline{\lambda}_{\max},\ \bar{\lambda}_{\max}(k)\}
\end{align}
where $\underline{\lambda}_{\max}$ represents the lower bound, which is set by the designer, and
\begin{align}\label{ma:calc_lambda_max}
\bar{\lambda}_{\max}(k)= \rho \frac{\sigma_{\max}(P(k))}{\sigma_{\min}(P(k))}\lambda_{\min}(k)
\end{align}
Further, the remaining forgetting factors $\lambda(k)$ are selected such that they do not have identical values within the range defined by $\lambda_{\min}(k)<\lambda(k)<\lambda_{\max}(k)$.
Note that if $\lambda_{\max}(k)$ becomes larger than $\lambda(k)$, it is necessary to select $\rho$ to a value closer to $1$.
In addition, if we choose $B(k) = \frac{1}{\sqrt{\lambda}}I$, it is equivalent to the conventional TLF-RLS.
These algorithms are summarized in Algorithm \ref{alg:ReEF}. Although the proposed method is simple, this is novel because it employs a two-layer structure to maintain the positive definiteness of the augmented regressor matrix and a forgetting factor to improve the condition number. 
In this framework, the hierarchical structure and the each inner/outer forgetting factor have a clear physical interpretation. Their effective integration enables adaptive updates that guarantee convergence to the true parameter values while maintaining high robustness against changes in system characteristics.

\begin{algorithm}
\caption{Reconfiguration EF algorithm}
\label{alg:ReEF}
\begin{algorithmic}
\State \textbf{Step 1}: Obtain $P(k)$

\State \textbf{Step 2}: Compute the singular value decomposition of $P(k)$ as
\[
P(k)=U(k)\Sigma_P(k)U^T(k)
\]
and identify the elements of the minimum and maximum eigenvalues.

\State \textbf{Step 3}: Set $\lambda_{\min}(k)$ and $\rho$

\State \textbf{Step 4}: Compute $\lambda_{\max}(k)$ as
\[
{\lambda}_{\max}(k)= \min\{\underline{\lambda}_{\max},\ \bar{\lambda}_{\max}(k)\}
\]
where $\underline{\lambda}_{\max}$ represents the lower bound and
\[
\bar{\lambda}_{\max}(k)= \rho \frac{\sigma_{\max}(P(k))}{\sigma_{\min}(P(k))}\lambda_{\min}(k)
\]
\State \textbf{Step 5}: If $\displaystyle\max_{1\leq i\leq n-1}(\lambda_i(k)) \geq \lambda_{n}(k)$,\\ where $\lambda_{n}(k)\triangleq \lambda_{\max}(k)$, return to \textbf{Step 3}.
\State \textbf{Step 6}: Set the eigenvalue matrix $\Lambda(k)$ according to the length of the elements of the covariance matrix, e.g.,
\[
\Lambda(k) = \mathrm{diag}\left\{\sqrt{\lambda_{\min}},\ \sqrt{\lambda},\ \sqrt{\lambda_{\max}}\right\},\ \lambda_{\min}< \lambda < \lambda_{\max}
\]

\State \textbf{Step 7}: Generate the forgetting matrix $B(k)$ as
\[
B(k)=U(k)\Lambda^{-1}(k)U^T(k)
\]
\State \textbf{Step 8}: Set $k \leftarrow k+1$ and return to \textbf{Step 1}
\end{algorithmic}
\end{algorithm}

\section{Stability analysis}
In this section, we present the boundedness of the information matrix for TLF-RLS using both the EF and ReEF algorithms.
Subsequently, we discuss the improvement in the condition number of the covariance matrix achieved by TLF-RLS using the ReEF algorithm. Finally, we demonstrate the stability of the parameter estimation error for the proposed methods.
The following result establishes the boundedness of the information matrix.
\begin{proposition}\cite{TLF-RLS}\label{prop:TRLS_R_Boundedness}
Consider Algorithm \ref{alg:gene_DF} and the update law \eqref{ma:law_P} of the covariance matrix. Assume that the regressor vector $\phi(k)$ satisfies the FE condition.
For all $k\geq k_e$, an upper bound $\alpha_{R} \in \mathbb{R}_{>0}$ and a lower bound $\beta_{R} \in \mathbb{R}_{>0}$ exist on the information matrix $R(k)$ such that
\begin{align}\label{ma:PE_Omega}
\alpha_{R} I < R(k) < \beta_{R} I.
\end{align}
where $\alpha_R = \alpha_{\Phi}$ and $\beta_R = \sigma_{\max}(R(\delta)) + \frac{1}{1-\lambda}\beta_{\Phi}$.
\end{proposition}
\begin{proof}
See \cite{TLF-RLS} or Appendix.
\end{proof}
\noindent It can be immediately shown that the boundedness of both the information matrix and the covariance matrix in TLF-RLS using ReEF is guaranteed.
\begin{corollary}\label{coro:boundedneses}
Consider Algorithms \ref{alg:gene_DF}, \ref{alg:ReEF} and the update law \eqref{ma:ReEF_P} of the covariance matrix. Assume that the regressor vector $\phi(k)$ satisfies the FE condition.
For all $k\geq k_e$, the information matrix $R(k)$ remains bounded.
\end{corollary}

Next, we evaluate the upper bounds of the condition number for the information matrix to examine whether the ReEF algorithm suppressed the estimation windup more effectively.
Here, we introduce the following two assumptions.
The eigenvalue decomposition of an information matrix $R(k)$ is expressed as 
\begin{align}\label{ma:Lambda_R}
R(k)=U(k)\Sigma_r(k)U^{T}(k)
\end{align}
where $\Sigma_r(k)=\mathrm{diag}\{r_1,\ \ldots,\ r_n\}$.

\begin{assumption}\label{ass:ReEF_stability}
There exists $0 < \rho < 1$ that satisfies 
\begin{align}
\frac{\lambda_{\max}}{\lambda_{\min}}\leq \rho \frac{r_{\max}}{r_{\min}}
\end{align}
where $\lambda_{\max},\ \lambda_{\min}$ represent the maximum and minimum forgetting factors in ReEF, respectively.
\end{assumption}

\begin{assumption}\label{ass:ReEF_stability2}
Let the elements of the eigenvalue matrix be defined in the descending order as $\lambda_1, \lambda_2, \ldots, \lambda_n$.
In that case, we assume that they are independently designed as 
\begin{align}
\lambda_1 > \lambda_2 > \cdots > \lambda_n
\end{align}
\end{assumption}

\begin{remark}
Assumption \ref{ass:ReEF_stability} is a practical or mild assumption.
Under non-PE condition, the condition number in EF-RLS or TLF-RLS is significantly greater than $1$. Therefore, a sufficiently wide range of forgetting factors in $(0,1)$ can be selected. The condition number of the information matrix improves significantly under PE condition by quasi-white noise, which makes it difficult to satisfy this condition. In such cases, it is reasonable to employ EF for uniform forgetting.
Moreover, Assumption \ref{ass:ReEF_stability2} is also reasonable because the forgetting factor is a design parameter. Hence, this condition can be readily satisfied.
\end{remark}
In both EF and ReEF algorithms, the update law of the information matrix can be expressed as 
\begin{align}\label{ma:ReEF_info}
R(k+1)=\bar{R}(k)+\Phi^2(k)
\end{align}
where matrix $\bar{R}(k)$ is obtained using the forgetting factor algorithm to $R(k)$.
Then, the following proposition holds.
\begin{proposition}\label{prop:cond}
Consider the parameter update law \eqref{ma:ReEF_theta} based on Algorithm \ref{alg:ReEF} under Assumptions \ref{ass:ReEF_stability} and \ref{ass:ReEF_stability2}.
Then, there exists $a\in(0,1)$ such that the following condition holds.
\begin{align}
\kappa(\bar{R}(k))\leq a \kappa(R(k))
\end{align}
where $\kappa(\cdot)$ represents the condition number.
\end{proposition}
\begin{proof}
For simplicity, we arrange the eigenvalues of the information matrix and forgetting factors in the descending order. According to \eqref{ma:Lambda_R} and Assumption \ref{ass:ReEF_stability2}, the forgetting factors are distinct.
\begin{align}
r_1\geq r_2\geq \cdots\geq r_n,\ \lambda_1 > \lambda_2 > \cdots > \lambda_n\label{ma:rlambda}
\end{align}
In this case, the eigenvalues of the information matrix and corresponding forgetting factors can be expressed in the product form $\lambda_{n+1-i} r_i$ ($i=1,\ldots,\ n$).
The objective is to satisfy the following conditions.
\begin{align}
\exists\ a\in(0,1)\ \mathrm{s.t.}\ \max_{i,\ j=1,\ldots,\ n}\frac{\lambda_{n+1-i}r_i}{\lambda_{n+1-j}r_j}\leq a\frac{r_1}{r_n}\label{ma:goal_cond}
\end{align}
We consider two cases: (a) $i < j$ and (b) $i > j$. Note that the case $i = j$ is excluded. \\
\textit{(a) $i < j\ \ $}From \eqref{ma:rlambda}, it is clear that $r_i \geq r_j$ and $\frac{\lambda_{n+1-i}}{\lambda_{n+1-j}}<1$ holds.
Then, we have
\begin{align}
\frac{\lambda_{n+1-i}r_i}{\lambda_{n+1-j}r_j}\leq \frac{\lambda_{n+1-i}}{\lambda_{n+1-j}}\frac{r_1}{r_n}=\gamma\frac{r_1}{r_n}
\end{align}
where $\gamma\triangleq \frac{\lambda_{n+1-i}}{\lambda_{n+1-j}}\in (0,1)$.\\
\textit{(b) $i > j\ \ $} From \eqref{ma:rlambda}, it is clear that $\frac{r_i}{r_j}\leq 1$ and $\frac{\lambda_{n+1-i}}{\lambda_{n+1-j}}\leq \frac{\lambda_1}{\lambda_n}$ holds.
Then, from Assumption \ref{ass:ReEF_stability}, we have 
\begin{align}
\frac{\lambda_{n+1-i}r_i}{\lambda_{n+1-j}r_j}\leq \frac{\lambda_{n+1-i}}{\lambda_{n+1-j}} \leq \frac{\lambda_1}{\lambda_n}\leq \rho\frac{r_1}{r_n}
\end{align}
Therefore, from (a) and (b), there exists an $\alpha=\max\{\gamma,\ \rho\}\in(0,1)$ that satisfies \eqref{ma:goal_cond}.
\end{proof}

\begin{theorem}\label{thm:cond}
Consider the parameter update law \eqref{ma:ReEF_theta} based on Algorithm \ref{alg:ReEF} under Assumptions \ref{ass:ReEF_stability} and \ref{ass:ReEF_stability2}.
Then, the upper bound on the condition number of the covariance matrix $\kappa_{\mathrm{ReEF}}(R(k+1))$ is uniformly ultimately bounded within the region satisfying Assumption \ref{ass:ReEF_stability}.
\end{theorem}
\begin{proof}
First, we derive the upper bound on the condition number of the covariance matrix.
Using Weyl's inequality (see Lemma \ref{lem:weyl} in Appendix), the update law of the information matrix \eqref{ma:ReEF_info} is evaluated as
\begin{align}
\sigma_{\max}(R(k+1))\leq\sigma_{\max}(\bar{R}(k))+\sigma_{\max}(\Phi^2(k))
\end{align}
where $\sigma_{\min}(\cdot),\ \sigma_{\max}(\cdot)$ represent the minimum and maximum eigenvalues, respectively.
In addition, we have
\begin{align}
\sigma_{\min}(\bar{R}(k))<\sigma_{\min}(R(k+1))
\end{align}
for the positive definiteness of $\Phi(k)$.
Then, we consider the condition number and its upper bound as
\begin{align}  
\kappa_{\mathrm{ReEF}}(R(k+1))  
&= \frac{\sigma_{\max}(R(k+1))}{\sigma_{\min}(R(k+1))} \\  
&< \frac{\sigma_{\max}(\bar{R}(k)) + \sigma_{\max}(\Phi^2(k))}{\sigma_{\min}(\bar{R}(k))} \\  
&= \kappa_{\mathrm{ReEF}}(\bar{R}(k)) + \frac{\sigma_{\max}(\Phi^2(k))}{\sigma_{\min}(\bar{R}(k))}
\label{ma:upper_kappa}
\end{align}
From Proposition \ref{prop:cond}, under Assumptions
\ref{ass:ReEF_stability} and \ref{ass:ReEF_stability2}, it follows that $\kappa(\bar{R}(k+1)) \leq \alpha \kappa(R(k))$.
Hence, the upper bound of the condition number can be rewritten as
\begin{align}
\kappa_{\mathrm{ReEF}}(R(k+1)) < \alpha\kappa_{\mathrm{ReEF}}(R(k)) + \frac{\sigma_{\max}(\Phi^2(k))}{\sigma_{\min}(\bar{R}(k))}\label{ma:upper_ReEF_R}
\end{align}
The second term on the right-hand side is bounded by Theorem \ref{thm:PE} and Proposition \ref{prop:TRLS_R_Boundedness}.
We define its positive upper bound $M\geq 0$ as ${\sigma_{\max}(\Phi^2(k))}/{\sigma_{\min}(\bar{R}(k))}\leq M$.
Subsequently, under Assumption \ref{ass:ReEF_stability}, we have
\begin{align}
\kappa_{\mathrm{ReEF}}(R(k))< \alpha^{k} \kappa_{\mathrm{ReEF}}(R(0))+\frac{M}{1-\alpha}(1-\alpha^k)
\end{align}
Hence, we have
\begin{align}
\kappa_{\mathrm{ReEF}}(R(k))<\max\left\{\kappa_{\mathrm{ReEF}}(P(0)),\ \frac{M}{1-\alpha}\right\}\label{ma:Kappa}
\end{align}
Therefore, $\kappa_{\mathrm{ReEF}}(R(k))$ is uniformly ultimately  bounded in the region satisfying $\frac{\lambda_{\max}}{\lambda_{\min}}\leq\rho\frac{r_{\max}}{r_{\min}}$.
\end{proof}

\begin{remark}
In TLF-RLS with the EF algorithm, similar to \eqref{ma:upper_kappa} and the fact that $\kappa(\bar{R}(k+1)) = \kappa(R(k))$ caused by the uniform forgetting property of EF, the upper bound of the condition number of the information matrix can be expressed as 
\begin{align}
\kappa_{\mathrm{EF}}(R(k+1)) < \kappa_{\mathrm{EF}}(R(k)) + \frac{\sigma_{\max}(\Phi^2(k))}{\lambda r_{\min}(k))}
\end{align}
Hence, its uniform ultimate boundedness is not guaranteed because there are no terms that reduce the upper bound of the condition number compared to the ReEF algorithm. 
Note that its boundedness is guaranteed by Proposition \ref{prop:TRLS_R_Boundedness}.
\end{remark}

\begin{remark}
As mentioned earlier, $\lambda_{\max}(k)$ can be determined automatically, for example, using \eqref{ma:calc_lambda_max}.
This is designed to satisfy Assumption \ref{ass:ReEF_stability}. The value of $\rho$ depends on $\alpha$, which determines the ultimate boundedness, and therefore, it must be selected from the interval $(0,1)$.
\end{remark}

Finally, the stability of the parameter error of the proposed method is proved via Lyapunov stability analysis based on the boundedness of the information matrix.
\begin{theorem}\label{thm:parameter}
Consider Algorithms \ref{alg:gene_DF}, \ref{alg:ReEF} and the parameter update law \eqref{ma:ReEF_theta}.
Assume that the regressor vector $\phi(k)$ satisfies the FE condition.
Then, parameter error $\tilde{\theta}(k)$ globally exponentially converges to zero independent of the forgetting factor of the ReEF algorithm.
\end{theorem}
\begin{proof}
Consider the following Lyapunov function candidate.
\begin{align}\label{ma:Lyap_candidate}
V(k)=\tilde{\theta}^T(k)R(k)\tilde{\theta}(k)
\end{align}
The information matrix is bounded based on Proposition \ref{prop:TRLS_R_Boundedness} and Corollary \ref{coro:boundedneses}. Therefore, it follows that, for some $\alpha_R,\ \beta_R\in\mathbb{R}_{>0}$
\begin{align}
\alpha_R \|\tilde{\theta}(k)\|^2 \leq V(k) \leq \beta_R \|\tilde{\theta}(k)\|^2
\end{align}
regardless of the PE condition.
The modified covariance matrix $L(k)$ obtained using the forgetting matrix $B(k)$ is positive definite for all $k \ge 0$, and therefore, this matrix can be transformed as 
\begin{align}
&P(k+1) = [I-L(k)\Phi(k)\bar{N}^{-1}(k)\Phi(k)]L(k)\\
&\iff I-L(k)\Phi(k)\bar{N}^{-1}(k)\Phi(k) = P(k+1)L^{-1}(k)
\end{align}
Then, the parameter error can be rewritten as
\begin{align}
\tilde{\theta}(k+1) &= \tilde{\theta}(k) - L(k)\Phi(k)\bar{N}^{-1}(k)\Phi(k)\tilde{\theta}(k) \notag \\
&= [I - L(k)\Phi(k)\bar{N}^{-1}(k)\Phi(k)]\tilde{\theta}(k)\label{ma:LLL}\\
&= P(k+1)L^{-1}(k)\tilde{\theta}(k)\label{ma:LLL2}.
\end{align}
On the one hand, using \eqref{ma:LLL} and \eqref{ma:LLL2}, the upper bound for a candidate Lyapunov function can be evaluated as
\begin{align}
V(k+1) &= \tilde{\theta}^T(k+1)L^{-1}(k)\tilde{\theta}(k) \notag \\
&= \tilde{\theta}^T(k)[I - L(k)\Phi(k)\bar{N}^{-1}(k)\Phi(k)]^TL^{-1}(k)\tilde{\theta}(k) \notag \\
&\leq \tilde{\theta}^T(k)L^{-1}(k)\tilde{\theta}(k)\label{ma:Lyap_temp}
\end{align}
On the other hand, we have
\begin{align}
L^{-1}(k)&=B^{-T}(k)R(k)B^{-1}(k)\notag\\
&=U(k)\Lambda(k)U^{T}(k)R(k)U(k)\Lambda(k)U^T(k)\notag\\
&\leq \sigma_{\max}(\Lambda(k))^2R(k)\notag\\
&=\lambda_{\max}(k)R(k)\label{ma:lambda_temp}
\end{align}
Therefore, from \eqref{ma:Lyap_temp} and \eqref{ma:lambda_temp}, we have
\begin{align}
V(k+1)&\leq \lambda_{\max}(k)\tilde{\theta}^T(k)R(k)\tilde{\theta}(k)\notag\\
& = \lambda_{\max}(k)V(k)\notag\\
&\leq \lambda^\prime_{\max}V(k)\label{ma:Lyap_fin}
\end{align}
where $\lambda^\prime_{\max}\triangleq \max_k({\lambda_{\max}(k)})\in(0,1)$
Therefore, the global exponential stability of the parameter error is obtained.
\end{proof}

\begin{remark}
In the conventional EF-RLS method, the existence of a lower bound for the information matrix cannot be guaranteed if the regressor vector does not satisfy the PE condition. In this case, the Lyapunov function candidate in \eqref{ma:Lyap_candidate} cannot be bounded, and thus, the stability of the parameter error cannot be guaranteed.
If the forgetting factor of EF $\lambda$ is set close to $1$, the system may not remain unstable. 
However, tuning is difficult because the forgetting factor $\lambda$ is highly sensitive. Furthermore, as shown in \eqref{ma:Lyap_fin}, this is a key parameter for determining the convergence rate of the parameter.
Therefore, it should ideally be set to a value sufficiently close to zero.
The proposed method, which is TLF-RLS using the EF and ReEF algorithms, guarantees the boundedness of the information matrix independently of the value of the forgetting factor.
Hence, the convergence rate of the parameter can be tuned at the design phase by the designer, which enhances practical applicability.
\end{remark}

\begin{remark}
In TLF-RLS using the EF algorithm, the upper bound of the Lyapunov function satisfies $V(k+1)\leq \lambda V(k)$. Therefore, because $\lambda^\prime_{\max} \ge \lambda$, the convergence rate of the parameter for TLF-RLS using ReEF is slower than that for TLF-RLS using the EF algorithm.
However, this estimate is conservative because the forgetting rate differs across directions owing to matrix-based forgetting.
According to Theorem \ref{thm:cond}, the ReEF algorithm improves the condition number, suppressing the estimation windup. Hence, although there is a trade-off between the parameter convergence rate and suppression of the estimation windup, the proposed method relaxes the trade-off.
\end{remark}

\section{Simulation}
Six adaptive parameter estimation laws are compared from three perspectives to evaluate the proposed method: 1) parameter convergence rate for time-invariant and time-varying systems with parameter jumps, 2) effect of the outer layer forgetting factor in TLF-RLS, and 3) estimation windup suppression by modifying the condition number using the ReEF algorithm.
The following three simulations were conducted.\\
\textbf{Simulation 1)}$\ $Comparison of the parameter update laws of EF-RLS, DF-RLS, CL, DF-CL, and TLF-RLS.\\
\textbf{Simulation 2)}$\ $Evaluation of the convergence rate of TLF-RLS under five different outer layer forgetting factors $\lambda$.\\
   \textbf{Simulation 3)}$\ $Evaluation of the trade-off between the parameter convergence rate and estimation windup suppression using the EF and ReEF algorithms in TLF-RLS.


\subsection{Simulation condition}
We evaluated its performance on a benchmark problem using a mass-spring-damper system to illustrate the effectiveness of the proposed method. The differential equation normalized by $m$ in \cite{FF_matrix} is expressed as
\begin{align}
\ddot{x}(t)+\frac{b}{m}\dot{x}(t)+\frac{k}{m}x(t) = \frac{1}{m}F(t),
\end{align}
where $m$, $k$, $b$, and $F(t)$ represent the mass, spring constant, damping coefficient, and external force, respectively.
Furthermore, we define the normalized external force by mass, i.e., the acceleration as $u(t)=\frac{F(t)}{m}$, and we use the discretized form obtained via a zero-order hold with a sampling time of $1$ s.
Furthermore, the spring constant and damping coefficient are set as follows for the three cases to consider parameter jumps for time varying system, with the mass fixed at $m = 5$ kg: (a) $k=1$ N/m, $b=1$ Ns/m, (b) $k=10$ N/m, $b=0.01$ Ns/m, and (c) $k=0.1$ N/m, $b=10$ Ns/m.
The ARX model of the mass–spring–damper system can be expressed as
\begin{align}
  y(k+1)=\phi^T(k)\theta
\end{align}
where $\phi(k)$ represents the regressor vector with the displacement $y(k) \in \mathbb{R}$ and the input $u(k) \in \mathbb{R}$.
In this study, the regressor vector is set as $\phi(k) = [y(k),\ y(k-1),\ u(k), u(k-1)]^T$, and the true parameter vector $\theta$ is given as 
\begin{enumerate}[(a)]
  \item $\theta=[1.6405,\ -0.8187,\ 0.4606,\ 0.430]^T$
  \item $\theta=[0.3116,\ -0.9980,\ 0.4218,\ 0.4215]^T$
  \item $\theta=[1.1267,\ -0.1353,\ 0.2834,\ 0.1482]^T$
\end{enumerate}
For all simulations, the initial values of the estimated parameters were set to $\hat{\theta}=[0,\ 0,\ 0,\ 0]^T$, and the initial values of the covariance matrix were set to $P(0)=1000I$.
Furthermore, the identification input was set to $u(k)=\sin(0.1k)$, which satisfies only the FE not PE condition.
Furthermore, for time-invariant systems, only the true value parameter (a) was used, whereas for time-varying systems, only the true value parameter (a) was used for $0\leq k <200$ steps, (b) $200\leq k <500$ steps, and (c) $500\leq k <1500$ steps.

\noindent\textbf{Simulation 1)} We compared EF-RLS, DF-RLS, CL, DF-CL, and TLF-RLS.
The EF or TLF-RLS in the outer forgetting factor was set to $\lambda = 0.99$, the DF or TLF-RLS in the inner forgetting factor was set to $\mu = 0.5$, and the forgetting factor in DF-CL was also set to $\mu = 0.5$.
The CL used an algorithm that maximizes the inverse of the condition number \cite{DCL}, and the number of stored data points was set to $4$, which was equal to the system order.
Furthermore, in the simulations of time-varying systems, the inner forgetting factor in TLF-RLS was set to $\mu=0.99$ to emphasize the robustness against changes in system characteristics, and DF-CL was set to be the same.
In TLF-RLS and DF-CL, the interpretation of the DF differed from that of the DF-RLS in that values closer to 1 corresponded to stronger forgetting. Therefore, the forgetting factor in DF-RLS was set to $\mu = 0.01$ to satisfy the same conditions.

\noindent\textbf{Simulation 2)} We evaluate the convergence rate of TLF-RLS using the outer forgetting factor. We fixed the DF in the inner forgetting factor at $\mu = 0.5$ and evaluated the identification performance by setting the EF in the outer forgetting factor to five values: $\lambda=0.99, 0.9, 0.8, 0.5, 0.01$.

\noindent\textbf{Simulation 3)} 
We compared TLF-RLS with EF and ReEF for the time-varying system.
The inner forgetting factor in TLF-RLS was fixed at $0.99$, whereas in TLF-RLS with the EF algorithm, the outer forgetting factor was set to three values: $0.99$, $0.5$, and $0.01$.
Furthermore, in TLF-RLS with the ReEF algorithm, the minimum forgetting factor $\lambda_{\min}(k)$ in the outer forgetting factors was set to $0.01$.
The lower bound of the maximum forgetting factor was set to $\bar{\lambda}_{\max} = 0.99$, and two values of the design parameter $\rho$ were examined: $0.01$ and $0.99$. Accordingly, $\lambda_{\max}(k)$ at each time step was computed using Algorithm \ref{alg:ReEF}.
Furthermore, the remaining forgetting factors were selected at intervals of $0.01$, starting from $\lambda_{\min}(k)$ to satisfy Assumption \ref{ass:ReEF_stability2}.

\subsection{Simulation results and discussion}
\textbf{Simulation 1)}
Figures \ref{fig:Iderror_LTI_case1} and \ref{fig:para_LTI_case1} show the identification and parameter errors for the time-invariant system for Case 1, respectively. 
Figure \ref{fig:Iderror_LTI_case1} shows that the identification errors for all methods approach zero, which implies that the adaptive identification algorithms achieve accurate identification.
Although the EF-RLS and DF-RLS methods reduce the identification error, as shown in Fig. \ref{fig:para_LTI_case1}, the parameter error does not decrease from a certain level because EF-RLS has a relatively large forgetting factor of $0.99$, which limits its performance. In addition, the boundedness of the covariance matrix cannot be guaranteed because the identification input does not satisfy the PE condition. Thus, setting a small forgetting factor can lead to instability.
Moreover, DF-RLS guarantees that the system is only uniformly stable in the sense of Lyapunov with respect to the parameter error. Although it remains stable under non-PE condition, no further performance improvement can be expected.
Although EF-RLS achieves a smaller parameter error compared to that of the CL and DF-CL during the transient response, the errors for the CL and DF-CL continue to decrease gradually and improve after approximately $1500$ steps compared to that for the  EF-RLS.
These results indicate that the parameter converge to the true value under non-PE condition. In addition, the DF-CL achieves a faster convergence rate than the CL method. This demonstrates the effectiveness of the DF algorithm and confirms that applying a forgetting factor to the augmented regressor matrix is more effective than applying the DF-RLS. 
Meanwhile, the TLF-RLS method achieves very fast convergence to the true parameter values because of its ability to set the outer layer forgetting factor to $0.01$. This is significant because conventional methods update the parameters using a single layer, whereas the proposed method employs a two-layer structure that ensures the PE condition through an augmented regressor matrix.

Figures \ref{fig:Iderror_LTV_case1} and \ref{fig:para_LTV_case1} 
show the identification and parameter errors for a time-varying system with parameter jumps under Case 1, respectively.
Figure \ref{fig:Iderror_LTV_case1} shows that the identification errors are sufficiently small for all methods, similar to that in the time-invariant case.
In the EF-RLS, a significant estimation windup occurs at approximately 200 and 500 steps, where parameter jumps occur. The norm of the parameter error remains large after 500 steps because of the very slow convergence rate. This behavior is attributed to the insufficient forgetting of past data.
Furthermore, the CL does not have a forgetting factor, and therefore, it cannot reduce the norm of the parameter error after changes in the system characteristics and exhibits oscillatory behavior.
Both DF-RLS and DF-CL suppress the estimation windup because of the advantages of the DF algorithm that preserves the positive definiteness of the information matrix.
However, as observed in the time-invariant case, DF-RLS has an almost constant norm of parameter error, whereas DF-CL can reduce this norm gradually because of a very slow convergence rate. Thus, neither method can achieve satisfactory performance.
However, DF-CL achieves performance similar to that of TLF-RLS after approximately 500 steps.
TLF-RLS demonstrates robustness against changes in system characteristics, as indicated by its fast parameter convergence and small parameter error norm, even after 200 and 500 steps.
However, a significant estimation windup occurs at the time steps where the system characteristics change, and its magnitude is comparable to that of EF-RLS.
 Thus, although the proposed method adopts the DF algorithm, the parameter convergence rate is contributed by the EF algorithm.
Further, Fig. \ref{fig:Omega_LTV_case3} shows the condition number with both the minimum and maximum eigenvalues of the augmented regressor matrix, respectively. In the steady-state response, the minimum eigenvalue is sufficiently large and the maximum eigenvalue is small, which results in a small condition number. These observations indicate that the EF algorithm is dominant.
Therefore, it is necessary to improve the outer forgetting factor. This issue is discussed in detail in Case 3 for TLF-RLS with the ReEF algorithm.

\noindent\textbf{Simulation 2)}
We evaluated the effect of the outer forgetting factor in TLF-RLS on the convergence rate.
Figures \ref{fig:Iderror_LTI_case2} and \ref{fig:para_LTI_case2} show the identification and parameter errors for Case 2, respectively. These figures indicate that the convergence rates of both the identification and parameter errors improve significantly as the outer forgetting factor approaches zero. According to Theorem \ref{thm:parameter}, this can be attributed to the outer forgetting factor determining the decreasing rate of the Lyapunov function.
Under non-PE condition, RLS cannot adopt the EF algorithm because setting a small EF factor easily leads to instability within a short period. In contrast, the proposed TLF-RLS remains globally exponentially stable with an arbitrary forgetting factor, which highlights its practical advantages.

\noindent\textbf{Simulation 3)}
We evaluated the suppression of the estimation windup achieved by TLF-RLS with the ReEF algorithm. Figures \ref{fig:iderror_LTV_case3} and \ref{fig:para_LTV_case3}
show the identification and parameter errors for Case 3, respectively.
For all TLF-RLS with EF and ReEF, the identification and parameter errors become sufficiently small in the steady-state response, achieving robustness against changes in system characteristics.
However, in TLF-RLS with the EF algorithm, an estimation windup occurs when the outer forgetting factor is set to small values such as $\lambda = 0.5$ or $\lambda = 0.01$, which confirms that the robustness during the transient performance is insufficient.
In the case of $\lambda = 0.99$ as the outer forgetting factor, the estimation windup is suppressed, particularly around 500 steps. However, the parameter convergence rate clearly decreases as indicated by Theorem \ref{thm:parameter}.
Further performance improvements are limited because TLF-RLS with the EF adopts a constant forgetting factor.
In contrast, TLF-RLS with ReEF demonstrates effective results through matrix-based forgetting.
In TLF-RLS with ReEF, the minimum forgetting factor is $\lambda_{\min}(k) = 0.01$, and the maximum forgetting factor $\lambda_{\max}(k)$ is updated, as shown in Fig. \ref{fig:FF_LTV_case3} based on two design parameters.
From this figure, we can observe that $\lambda_{\max}(k) \approx 0.7$ for $\rho = 0.01$ and $\lambda_{\max}(k) = 0.99$ for $\rho = 0.99$.
In both cases, the estimation windup is suppressed effectively, and its magnitude is reduced to approximately one-tenth of that in TLF-RLS with the EF algorithm.
Furthermore, the parameter convergence rate remains high, in particular, the transient convergence performance is not worse than that of TLF-RLS with the EF algorithm when the outer forgetting factor is set to $\lambda = 0.5$ or $0.01$.
These results can be interpreted in terms of the minimum and maximum eigenvalues and the condition numbers of the covariance matrices shown in Figs. \ref{fig:Pmin_LTV_case3}, \ref{fig:Pmax_LTV_case3}, and \ref{fig:cond_LTV_case3}.
In TLF-RLS with the EF algorithm, focusing on the minimum eigenvalue of the covariance matrix reveals that this value becomes very small when $\lambda = 0.99$. This confirms that the adaptive capability is significantly reduced.
Moreover, the minimum eigenvalue of the covariance matrix increases with a decrease in $\lambda$; thus, a small forgetting factor is necessary to improve the robustness against changes in the system characteristics.
The maximum eigenvalue of the covariance matrix increases with a decrease in the forgetting factor. A larger maximum eigenvalue of the covariance matrix leads to more sensitive parameter updates, which makes the estimation windup more likely. Thus, a smaller value is preferable. Therefore, in TLF-RLS with the EF algorithm, increasing the minimum eigenvalue and decreasing the maximum eigenvalue involves a trade-off.
In contrast, for TLF-RLS with the ReEF algorithm, the minimum eigenvalue is large and comparable to that of TLF-RLS with the EF at $\lambda = 0.01$ regardless of the value of $\rho$, whereas it remains small and comparable to that of TLF-RLS with the EF algorithm at $\lambda = 0.99$.
Although slight differences are observed based on the value of $\rho$, both the minimum and maximum eigenvalues are slightly larger for $\rho = 0.01$ than those for $\rho = 0.99$; however, this difference is negligible compared to that in TLF-RLS with the EF algorithm.
These results are evident from Fig. \ref{fig:cond_LTV_case3}; TLF-RLS with ReEF achieves a sufficiently small condition number, and it is reasonable to conclude that it exhibits high robustness against characteristics changes during both steady-state and transient response.
Further, its effectiveness is demonstrated by guaranteeing the uniform ultimate boundedness of the condition number using the ReEF algorithm, as shown in Theorem \ref{thm:cond}. Furthermore, compared to the DF-CL case shown in Fig. \ref{fig:para_LTV_case1}, TLF-RLS with the ReEF algorithm achieves a comparable level of estimation windup suppression. Moreover, this algorithm achieves a significantly faster parameter convergence rate than that of DF-CL.
Therefore, TLF-RLS with the ReEF algorithm simultaneously achieves fast parameter convergence and suppresses the estimation windup during the transient response, addressing the trade-off between these objectives.

\begin{figure}[ht]
\centering
\includegraphics[width=\columnwidth]{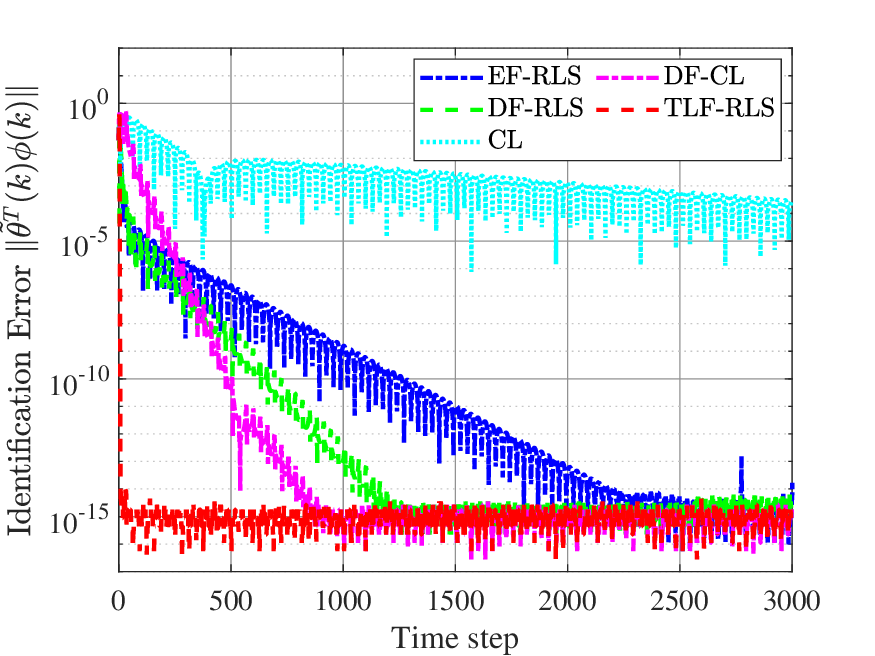}
\caption{Comparison of identification error for the LTI system under Case 1 (EF-RLS: $\lambda=0.99$, DF-RLS: $\mu=0.5$, CL: no forgetting, DF-CL: $\mu=0.5$, and TLF-RLS: $\lambda=0.01, \mu=0.5$)}
\label{fig:Iderror_LTI_case1}
\end{figure}

\begin{figure}[ht]
\centering
\includegraphics[width=\columnwidth]{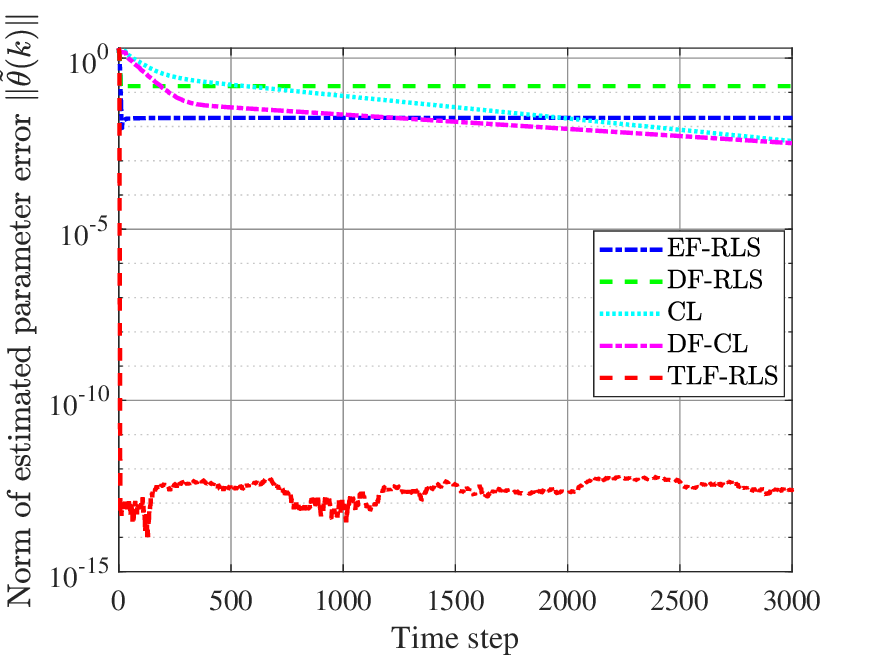}
\caption{Comparison of the parameter error for the LTI system under Case 1 (EF-RLS: $\lambda=0.99$, DF-RLS: $\mu=0.5$, CL: no forgetting, DF-CL: $\mu=0.5$, and TLF-RLS: $\lambda=0.01, \mu=0.5$)}
\label{fig:para_LTI_case1}
\end{figure}

\begin{figure}[ht]
\centering
\includegraphics[width=\columnwidth]{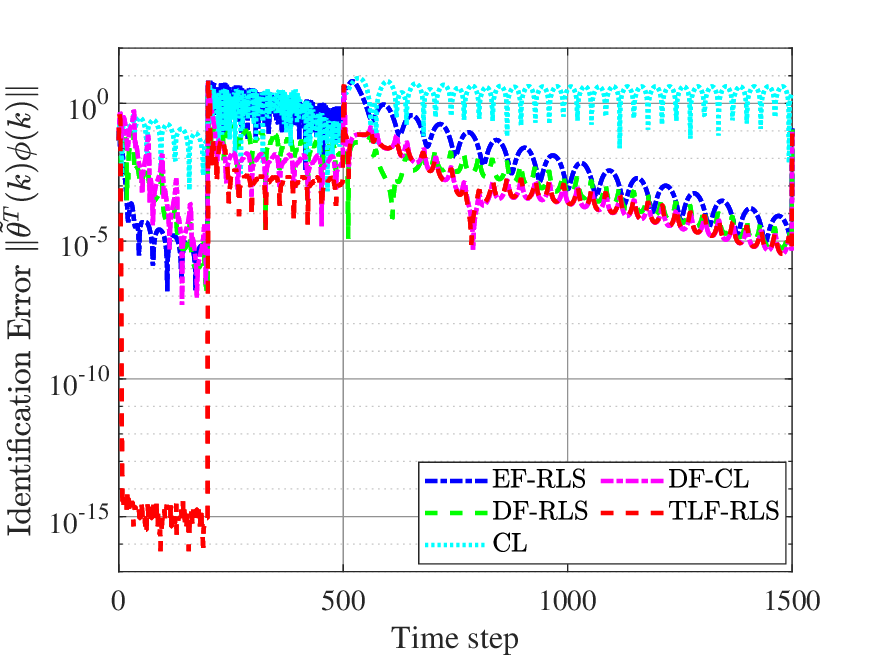}
\caption{Comparison of the identification error for the system with a parameter jump under Case 1 (EF-RLS: $\lambda=0.99$, DF-RLS: $\mu=0.01$, CL: no forgetting, DF-CL: $\mu=0.99$, and TLF-RLS: $\lambda=0.01, \mu=0.99$). The true value parameter (a) was used for $0\leq k <200$ steps, (b) $200\leq k <500$ steps, and (c) $500\leq k <1500$ steps)}
\label{fig:Iderror_LTV_case1}
\end{figure}

\begin{figure}[ht]
\centering
\includegraphics[width=\columnwidth]{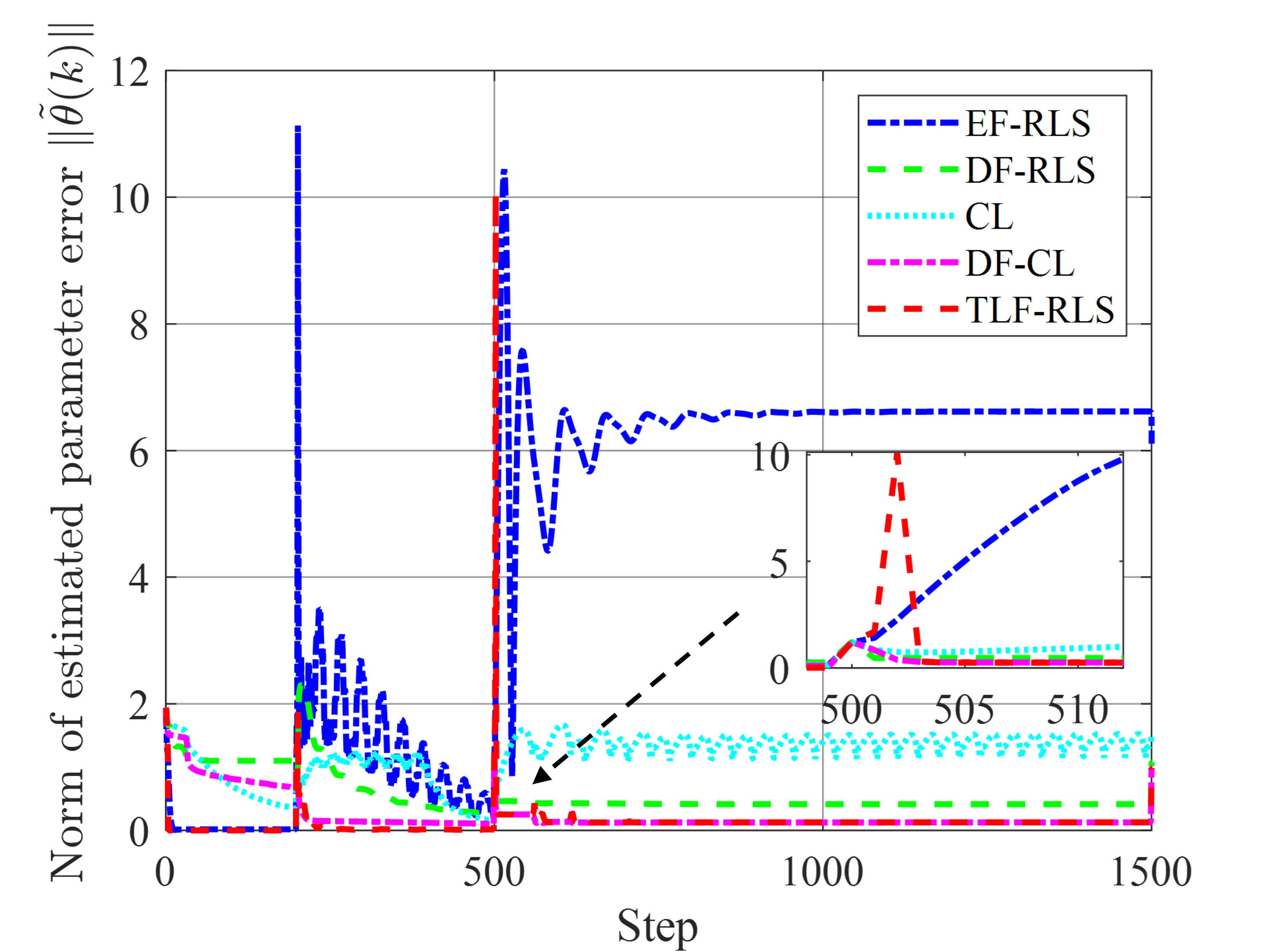}
\caption{Comparison of the parameter error for the system with parameter jump under Case 1 (EF-RLS: $\lambda=0.99$, DF-RLS: $\mu=0.01$, CL: no forgetting, DF-CL: $\mu=0.99$, and TLF-RLS: $\lambda=0.01, \mu=0.99$). The true value parameter (a) was used for $0\leq k <200$ steps, (b) $200\leq k <500$ steps, and (c) $500\leq k <1500$ steps)}
\label{fig:para_LTV_case1}
\end{figure}

\begin{figure}[ht]
\centering
\includegraphics[width=\columnwidth]{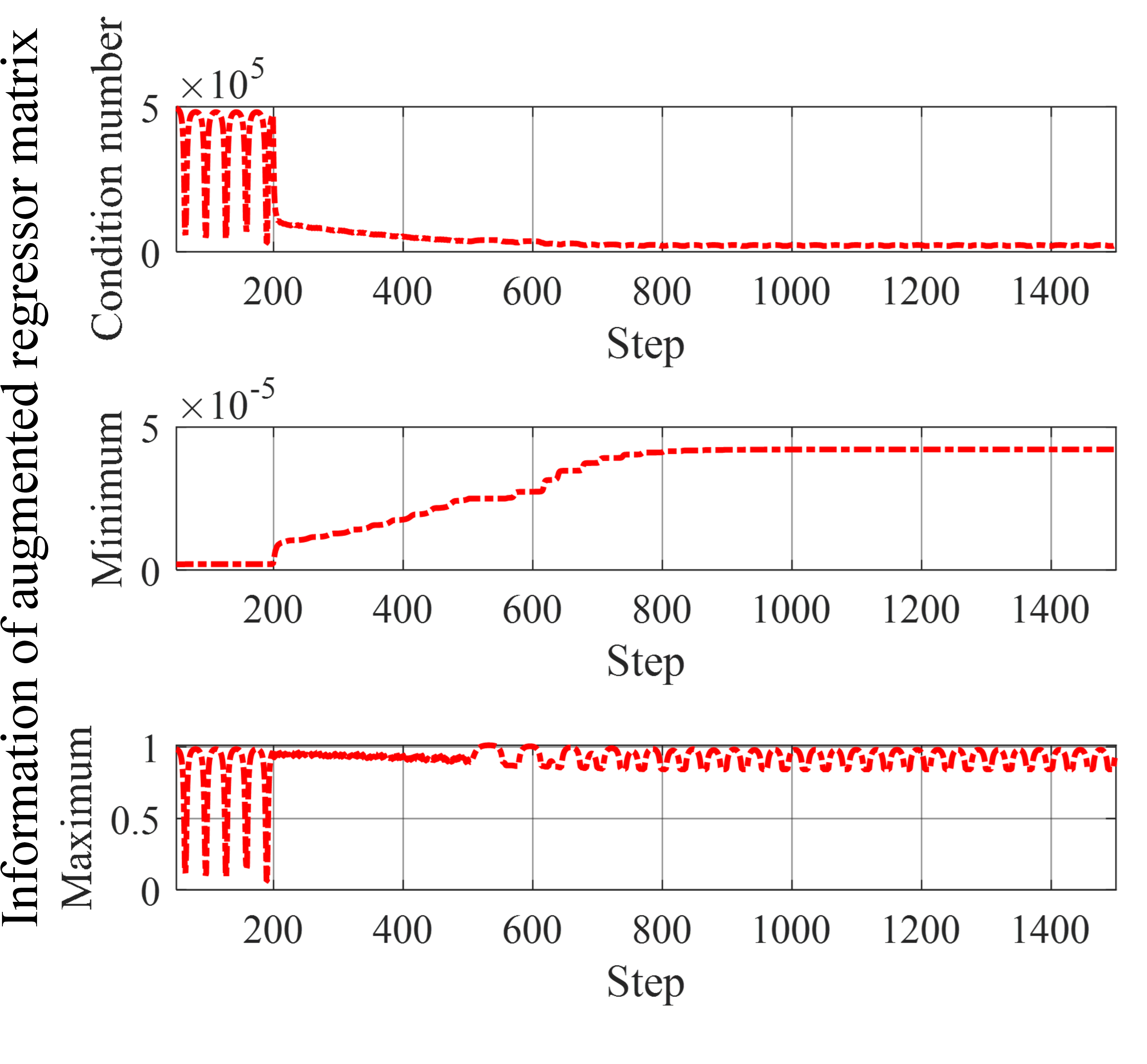}
\caption{Information of the augmented regressor matrix (top: condition number, middle: minimum eigenvalue, and bottom: maximum eigenvalue)}
\label{fig:Omega_LTV_case3}
\end{figure}

\begin{figure}[ht]
\centering
\includegraphics[width=\columnwidth]{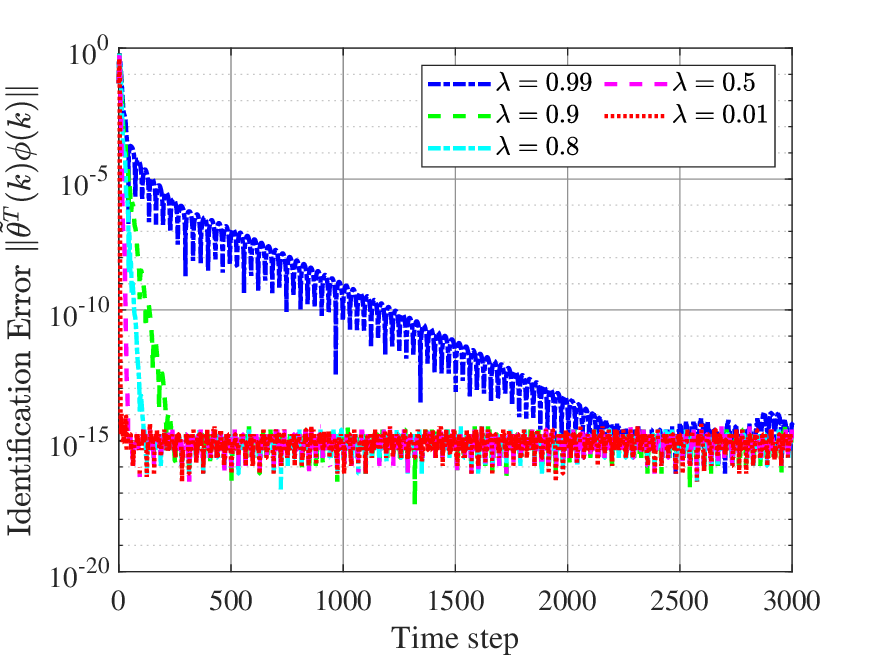}
\caption{Comparison of identification error for TLF-RLS with the EF algorithm ($\mu=0.99$, $\lambda=0.99, 0.9, 0.8, 0.5$ and $0.01$)}
\label{fig:Iderror_LTI_case2}
\end{figure}

\begin{figure}[ht]
\centering
\includegraphics[width=\columnwidth]{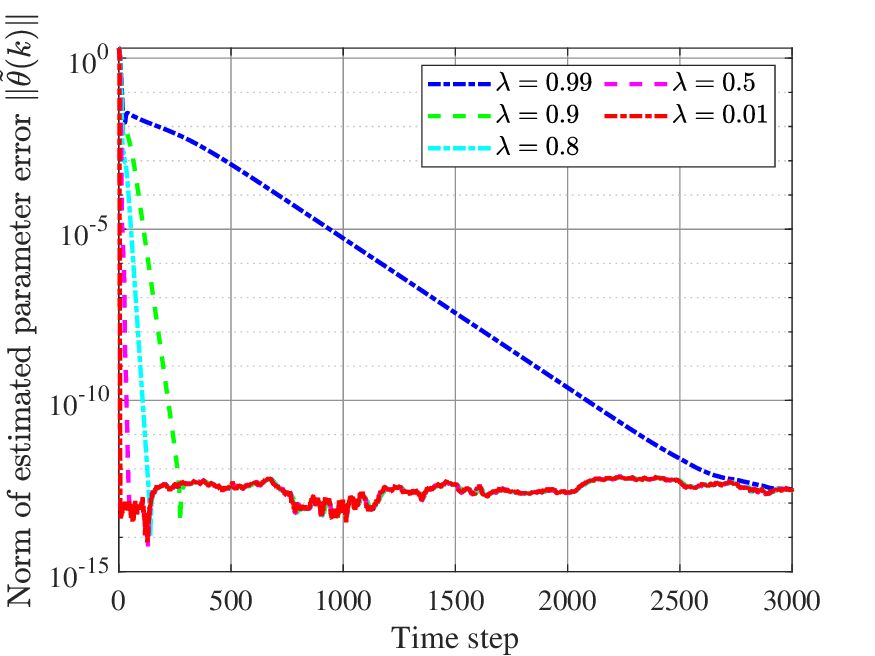}
\caption{Comparison of parameter error for TLF-RLS with the EF algorithm ($\mu=0.99$, $\lambda=0.99, 0.9, 0.8, 0.5$ and $0.01$)}
\label{fig:para_LTI_case2}
\end{figure}

\begin{figure}[ht]
\centering
\includegraphics[width=\columnwidth]{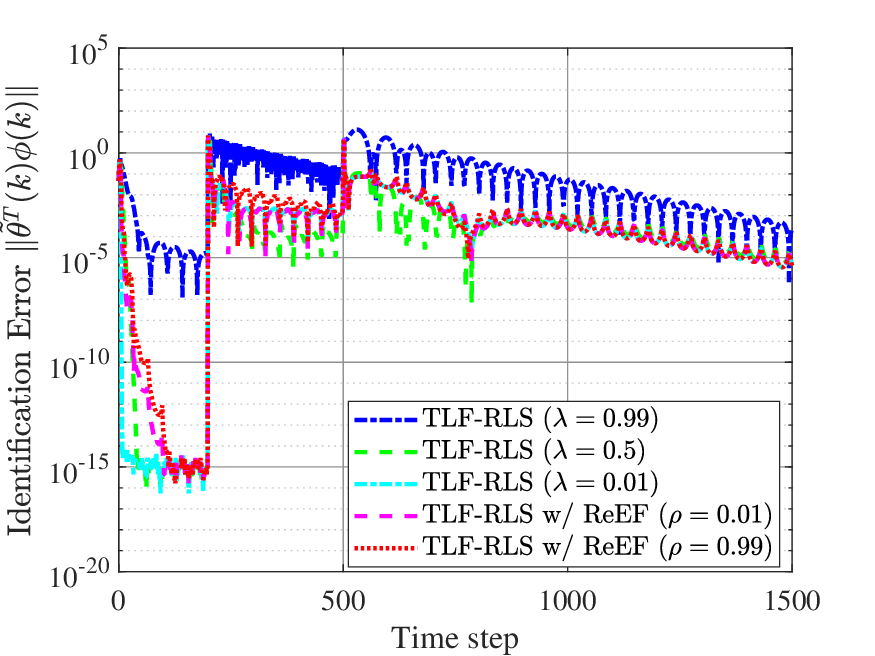}
\caption{Comparison of the identification error for TLF-RLS with EF and ReEF for the system with parameter jump under Case 3 ($\mu=0.99$; the true value parameter (a) was used for $0\leq k <200$ steps, (b) $200\leq k <500$ steps, and (c) $500\leq k <1500$ steps)}
\label{fig:iderror_LTV_case3}
\end{figure}

\begin{figure}[ht]
\centering
\includegraphics[width=\columnwidth]{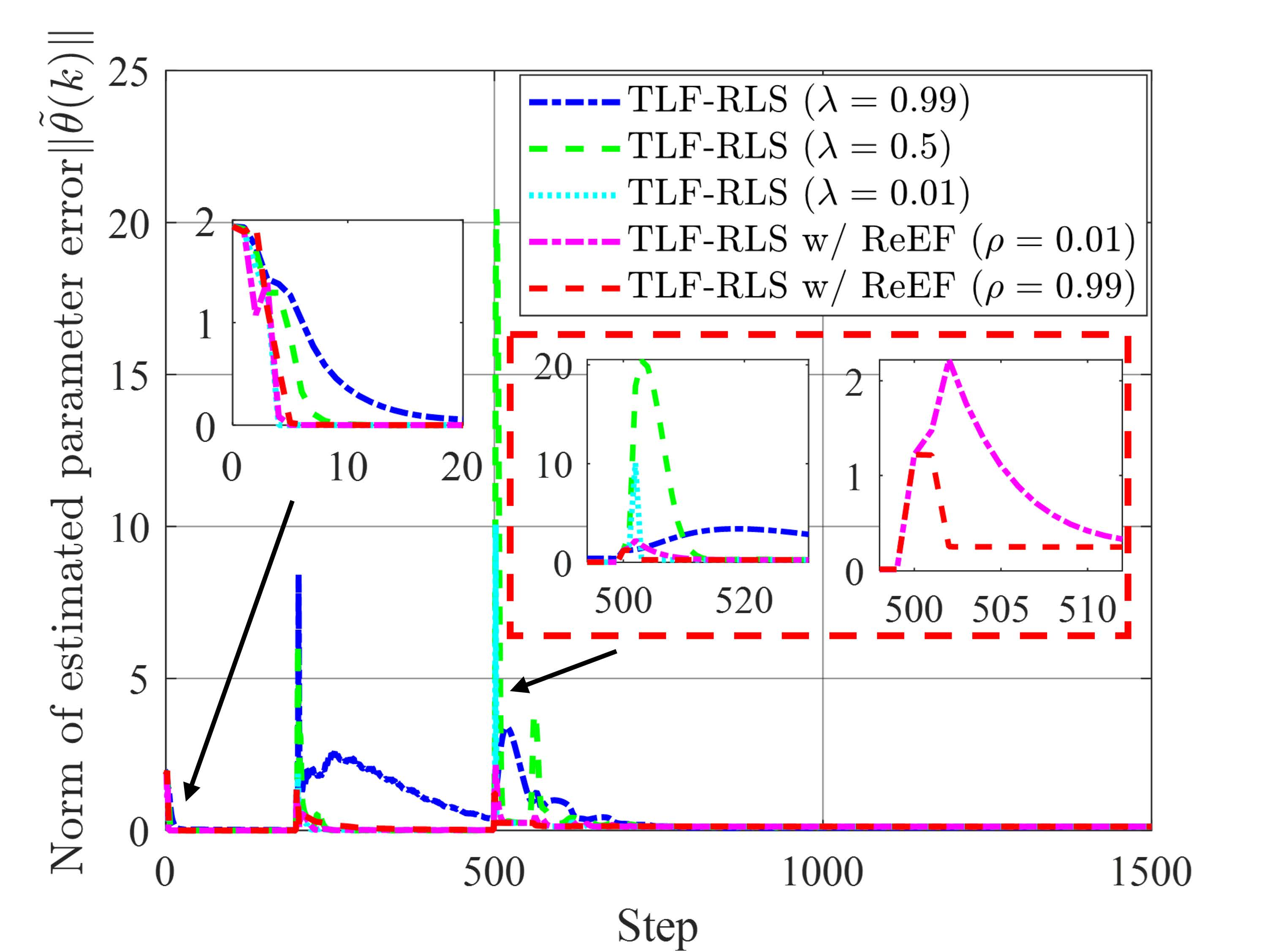}
\caption{Comparison of parameter error for TLF-RLS with EF and ReEF for the system with parameter jump under Case 3 ($\mu=0.99$; the true value parameter (a) was used for $0\leq k <200$ steps, (b) $200\leq k <500$ steps, and (c) $500\leq k <1500$ steps)}
\label{fig:para_LTV_case3}
\end{figure}
\begin{figure}[ht]
\centering
\includegraphics[width=\columnwidth]{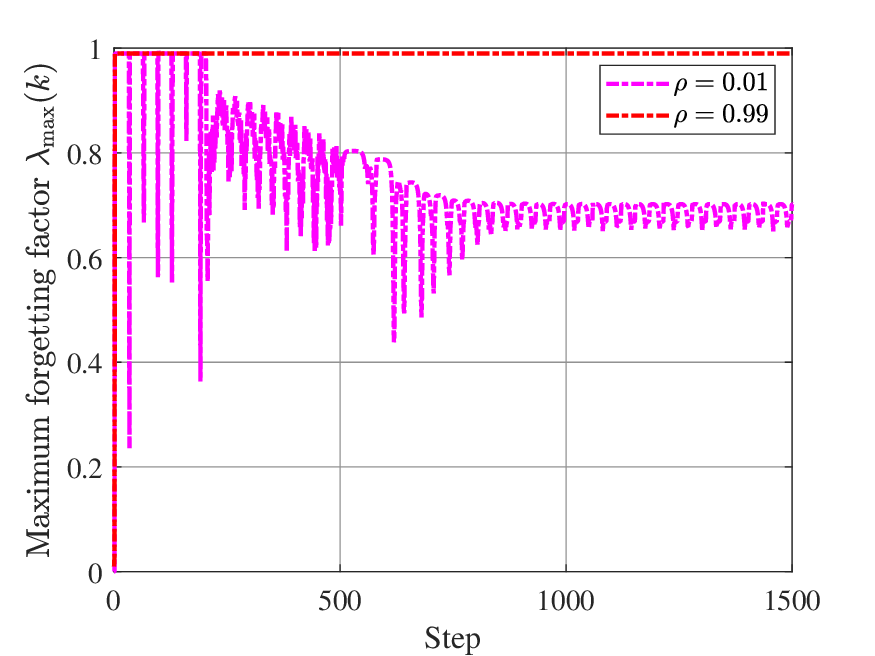}
\caption{Comparison of maximum forgetting factor $\lambda_{\max}(k)$ for $\rho = 0.01$ and $\rho = 0.99$}
\label{fig:FF_LTV_case3}
\end{figure}
\begin{figure}[ht]
\centering
\includegraphics[width=\columnwidth]{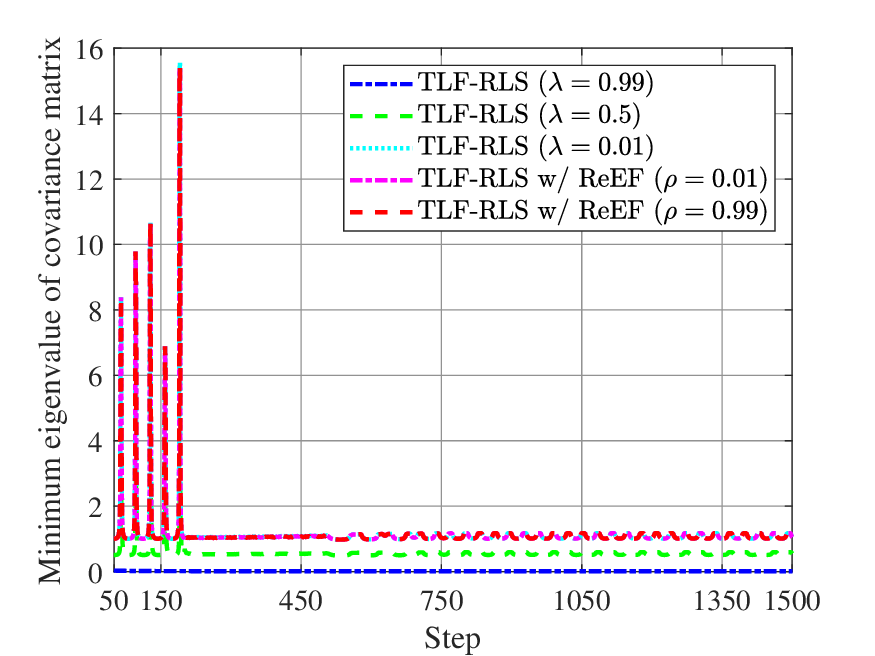}
\caption{Comparison of the minimum eigenvalue of the covariance matrix in the steady-state response for TLF-RLS with EF and ReEF for the system with parameter jump under Case 3 ($\mu=0.99$; the true value parameter was used for (a) $0\leq k <200$ steps, (b) $200\leq k <500$ steps, and (c) $500\leq k <1500$ steps)}
\label{fig:Pmin_LTV_case3}
\end{figure}
\begin{figure}[ht]
\centering
\includegraphics[width=\columnwidth]{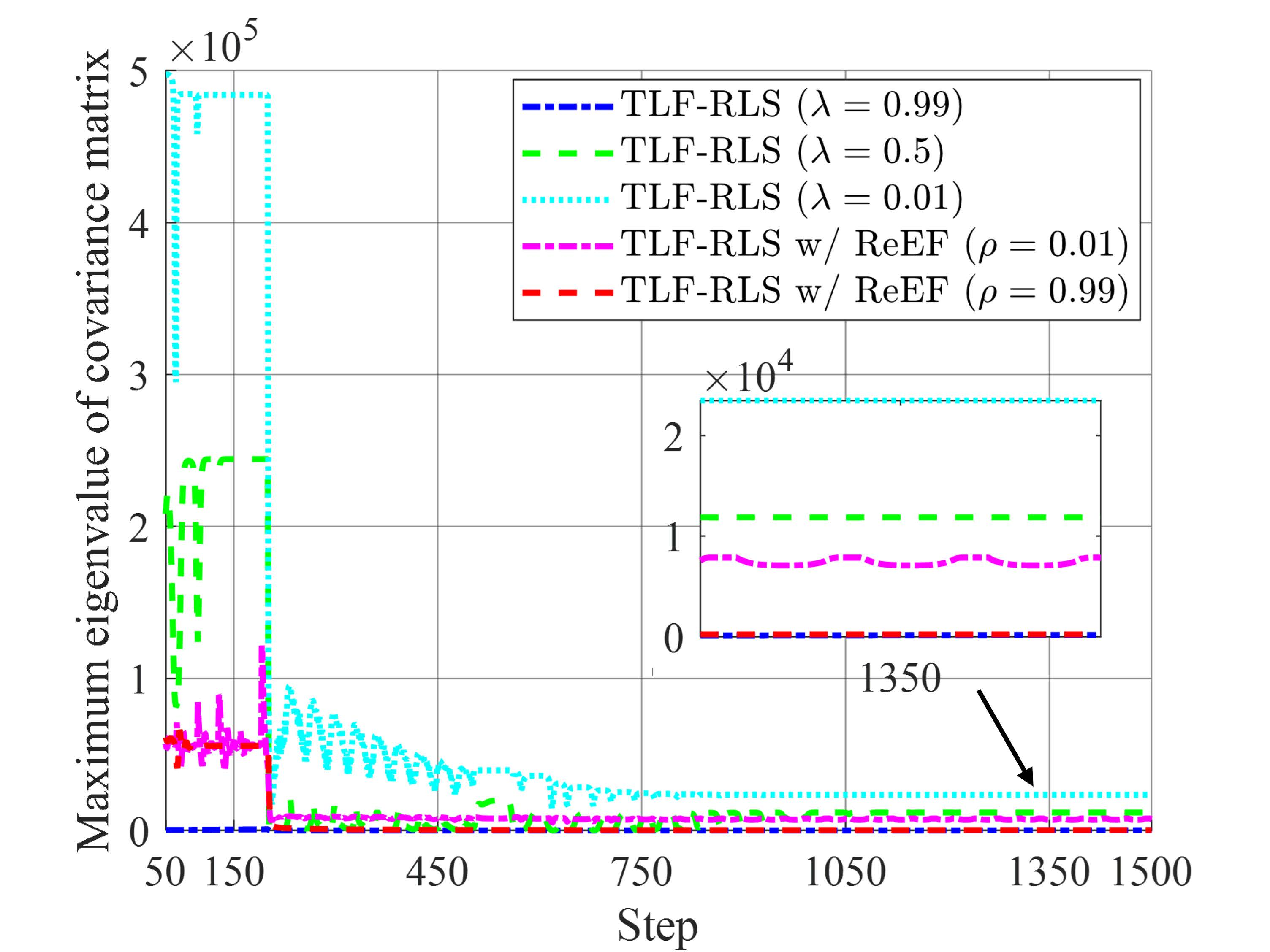}
\caption{Comparison of the maximum eigenvalue of the covariance matrix in the steady-state response for TLF-RLS with EF and ReEF for the system with parameter jump under Case 3 ($\mu=0.99$; the true value parameter was used for (a) $0\leq k <200$ steps, (b) $200\leq k <500$ steps, and (c) $500\leq k <1500$ steps)}
\label{fig:Pmax_LTV_case3}
\end{figure}
\clearpage
\begin{figure}[ht]
\centering
\includegraphics[width=\columnwidth]{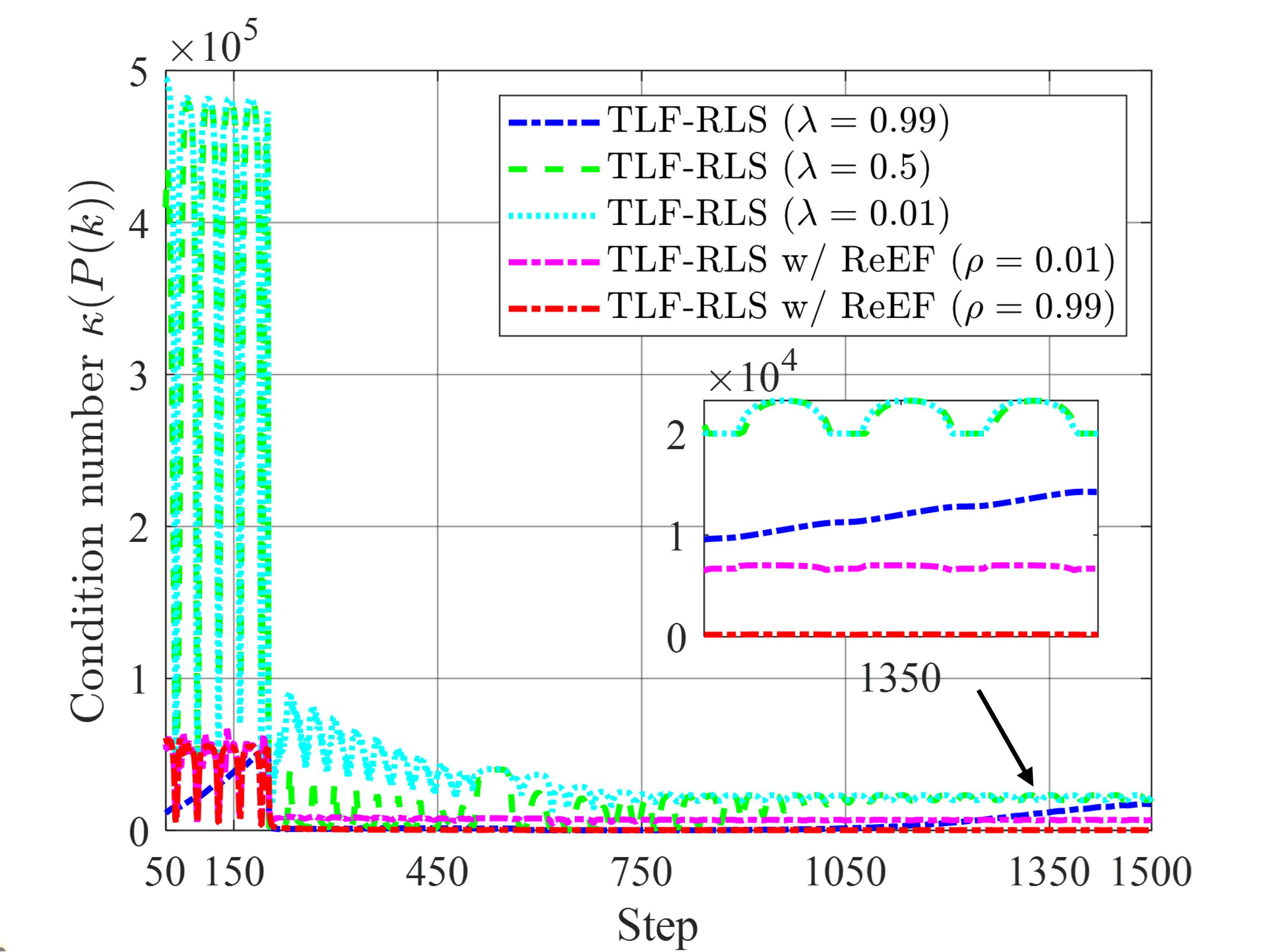}
\caption{Comparison of the condition number of the covariance matrix in the steady-state response for TLF-RLS with EF and ReEF for the system with parameter jump under Case 3 ($\mu=0.99$; the true value parameter was used for (a) $0\leq k <200$ steps, (b) $200\leq k <500$ steps, and (c) $500\leq k <1500$ steps)}
\label{fig:cond_LTV_case3}
\end{figure}

\section{Conclusion}
This paper proposed three novel adaptive identification algorithms under FE conditions based on an augmented regressor matrix that utilizes DF. A two-layer forgetting structure and proposed reconfiguration-based EF (ReEF) factor were introduced to address the trade-off between the parameter convergence rate and estimation windup suppression under non-PE condition.
Conventional methods such as EF-RLS and DF-RLS exhibit inherent limitations under non-PE condition; EF-RLS suffers from the instability or degraded performance, whereas DF-RLS ensures that the system is only stable in the sense of Lyapunov.
In contrast, adaptive identification using DF-based augmented regressor matrix such as DF-CL and TLF-RLS achieves the global exponential stability of parameter error for time-invariant systems and exhibits high robustness against changes in system characteristics in the steady-state response under non-PE condition.
The proposed two-layer forgetting structure enables an arbitrary selection of forgetting factors under non-PE condition, achieving a flexible tuning of the convergence rate. The ReEF algorithm improves the transient response after changes in the system characteristics while maintaining an estimation windup suppression comparable to that of conventional methods such as DF-RLS. 
Moreover, it guarantees the uniform ultimate boundedness of the condition number under mild assumptions. Numerical simulations demonstrated that TLF-RLS with the ReEF algorithm successfully achieved both fast parameter convergence and effective suppression of the estimation windup in both transient and steady-state responses under non-PE condition.
These results confirm that the proposed method effectively resolves the trade-off between convergence rate and robustness, providing a practical framework for adaptive identification in time-varying systems.
\bibliography{mybibfile}

\appendix
\section{Weyl's inequalities}
\begin{lemma}\label{lem:weyl}
Let $A, B \in \mathbb{C}^{n \times n}$ be Hermitian. Consequently, for all $i = 1, \ldots, n$,
\begin{equation}
\sigma_i(A) + \sigma_{\min}(B) \le \sigma_i(A + B) \le \sigma_i(A) + \sigma_{\max}(B).
\end{equation}
\end{lemma}
\begin{proof}
See Corollary 4.3.15 in \cite{matrix}. 
\end{proof}
\section{Proof of Proposition 2}
\begin{proof}
This proof can be established based on the analyses in \cite{FF_survey,EF-RLS,TLF-RLS}. The update law of the information matrix for TLF-RLS with the EF algorithm is expressed as
\begin{align}\label{ma:TRLS_update_R}
R(k) &= \lambda R(k-1) + \Phi^2(k-1)
\end{align}
For any $\sigma \geq 0$, and considering the sum over the interval $[\sigma, \sigma+\delta]$ on both sides, we can rewrite this as follows by Theorem \ref{thm:PE}.
\begin{align}\label{ma:TRLS_R_temp}
\sum_{k=\sigma}^{\sigma+\delta} R(k) 
&= \sum_{k=\sigma}^{\sigma+\delta} \left( \lambda R(k-1) + \Phi^2(k-1) \right) \notag \\
&\geq \sum_{k=\sigma}^{\sigma+\delta} \Phi^2(k-1) \geq \alpha_{\Phi} I.
\end{align}
From Proposition \ref{prop_DF} and (\ref{ma:TRLS_update_R}), it follows that $R(k) > \lambda R(k-1)$ holds for any $k \geq k_e$.
Therefore, (\ref{ma:TRLS_R_temp}) can be rewritten as 
\begin{align}
\sum_{i=0}^{\delta} \frac{1}{\lambda^i} R(\sigma+\delta) 
&> \sum_{k=\sigma}^{\sigma+\delta} R(k) \geq \alpha_{\Phi} I \notag \\
\Longleftrightarrow \quad
R(\sigma+\delta) &> \frac{\frac{1}{\lambda}-1}{\frac{1}{\lambda^{\delta+1}}-1}\alpha_{\Phi} I = \alpha_R I > 0
\end{align}
where $\alpha_R\triangleq{(\frac{1}{\lambda -1})}/{(\frac{1}{\lambda^{\delta+1}}-1)}$.
Next, we evaluate the upper bound of information matrix $R(k)$.
For any $k \geq k_e$, because $R(k+1) > \lambda R(k)$ holds, the following inequality can be obtained.
\begin{align}
\sum_{i=0}^{\delta} \lambda^i R(\sigma+1) < \sum_{i=0}^{\delta} R(\sigma+i+1),
\quad \forall \sigma \geq k_e.
\end{align}
Moreover, for any $k \geq k_e$, the following inequality also holds.
\begin{align}
R(\sigma+1)&<\frac{1-\lambda}{1-\lambda^{\delta+1}}\sum_{k=\sigma}^{\sigma+\delta}R(k+1)\nonumber\\
&= \frac{1-\lambda}{1-\lambda^{\delta+1}}
\left(\lambda\sum_{k=\sigma}^{\delta+\sigma} R(k)
  + \sum_{k=\sigma}^{\sigma+\delta} \Phi^{2}(k)
\right)\nonumber\\
&= \frac{1-\lambda}{1-\lambda^{\delta+1}}
\left(\lambda^\sigma \sum_{k=\sigma}^{\delta+\sigma} R(k-\sigma)+ \sum_{i=0}^{\sigma}\lambda^i\sum_{k=\sigma}^{\sigma+\delta} \Phi^{2}(k)
\right)
\end{align}
Using Theorem \ref{thm:PE}, we obtain
\begin{align}
R(\sigma+1)&<\frac{1-\lambda}{1-\lambda^{\delta+1}}\left(\lambda^{\sigma}\sum_{k=0}^{\delta}R(k)+\frac{1-\lambda^{\sigma}}{1-\lambda}\beta_{\Phi} I\right).\label{ma:TRLS_temp1}
\end{align}
Furthermore, from the properties of the information matrix update law, the following relationship holds.
\begin{align}
\lambda^{\sigma} \sum_{k=0}^{\delta} R(k)&\leq \lambda^{\sigma} \sum_{i=0}^{\delta} \frac{1}{\lambda^{i}}R(\delta)\nonumber\\
&=\lambda^{\sigma} \frac{1-\frac{1}{\lambda^{\delta+1}}}{1-\frac{1}{\lambda}} R(\delta)\label{ma:TRLS_temp2}
\end{align}
By substituting (\ref{ma:TRLS_temp2}) into (\ref{ma:TRLS_temp1}) and rearranging, we obtain
\begin{align}
R(\sigma+1) &< \frac{1-\lambda}{1-\lambda^{\delta+1}}
\left(
  \lambda^{\sigma} \frac{1 - \tfrac{1}{\lambda^{\delta+1}}}{1 - \tfrac{1}{\lambda}} R(\delta)
  + \frac{1-\lambda^{\sigma}}{1-\lambda} \beta_{\Phi} I
\right)\nonumber\\
&= \frac{1-\lambda}{1-\lambda^{\delta+1}}
\left(
  \lambda^{\sigma} \frac{1-\lambda^{\delta+1}}{\lambda^{\delta}(1-\lambda)} R(\delta)
  + \frac{1-\lambda^{\sigma}}{1-\lambda} \beta_{\Phi} I
\right)\nonumber\\
&\leq \lambda^{\sigma-\delta}R(\delta)+\frac{1-\lambda^{\sigma}}{1-\lambda^{\delta+1}}\beta_{\Phi} I\nonumber\\
&\leq R(\delta)+\frac{1}{1-\lambda^{\delta+1}}\beta_{\Phi} I\triangleq \beta_R I
\end{align}
where $\beta_R\triangleq \sigma_{\max}(R(\delta))+\frac{1}{1-\lambda^{\delta+1}}\beta_{\Phi}>0$.
Therefore, the following inequality holds for any $k \geq k_e$,
\begin{align}
\alpha_R I < R(k) < \beta_R I
\end{align}

Further, after time step $k \geq k_e$, we can set $\delta = 0$; this is because $\delta = k_e$ corresponds to the interval in Theorem \ref{thm:PE}, and the positive definiteness of the extended regressor matrix is always satisfied after time step $k_e$, eliminating the need for interval $\delta$.
Therefore, the following relationship holds.
\begin{align}
\alpha_R = \alpha_{\Phi}, \qquad 
\beta_R = \sigma_{\max}(R(\delta)) + \frac{1}{1 - \lambda}\beta_{\Phi}.
\end{align}
\end{proof}




\end{document}